\documentclass[12pt,one column]{IEEEtran}
\usepackage{setspace}
\pdfoutput=1
\usepackage{amsmath}
\usepackage{amssymb}
\usepackage{amsthm}
\usepackage{bbm}
\usepackage{dsfont}
\usepackage{cite}
\usepackage{color}
\usepackage{epsfig}

\usepackage{epsf}
\usepackage{rotating}
\usepackage{mathrsfs}
\usepackage{epsfig}
\usepackage{graphics}
\usepackage{enumerate}
\usepackage{euscript}
\usepackage{accents}
\DeclareMathAccent{\wtilde}{\mathord}{largesymbols}{"65}
\usepackage[T3,T1]{fontenc}
\DeclareSymbolFont{tipa}{T3}{cmr}{m}{n}
\DeclareMathAccent{\invbreve}{\mathalpha}{tipa}{16}
\usepackage[multiple]{footmisc}
\usepackage{anyfontsize}
\usepackage{t1enc}
\usepackage{amssymb}
\usepackage{cite}
\usepackage{color}
\usepackage{epsfig}
\usepackage{epsf}
\usepackage{rotating}
\usepackage{mathrsfs}
\usepackage{epsfig}
\usepackage{graphics}
\usepackage{subfigure}

\theoremstyle{plain}
\newtheorem{thm}{Theorem}
\newtheorem{lem}{Lemma}

\newtheorem{proposition}{Proposition}

\begin{document}
\vskip 1cm

\thispagestyle{empty} \vskip 1cm


\title{{On Interference Channels with \\Gradual Data Arrival }}
\author{ Kamyar Moshksar\\
\small Vancouver, BC, Canada} \maketitle

\begin{abstract}
We study memoryless interference channels with gradual data arrival in the absence of feedback. The information bits arrive at the transmitters according to independent and asynchronous~(Tx-Tx asynchrony) Bernoulli processes with average data rate $\lambda$. Each information source turns off after generating a number of $n$ bits. In a scenario where the transmitters are unaware of the amount of Tx-Tx asynchrony, we say $\epsilon$ is an \textit{achievable outage level} in the asymptote of large~$n$ if (i) the average transmission rate at each transmitter is $\lambda$ and (ii) the probability that the bit-error-rate at each receiver does not eventually vanish is not larger than~$\epsilon$. Denoting the infimum of all achievable outage levels by $\epsilon(\lambda)$, the contribution of this paper is an upper bound (achievability result) on $\epsilon(\lambda)$. The proposed method of communication is a simple block transmission scheme where a transmitter sends a random point-to-point codeword upon availability of enough bits in its buffer. 
Both receivers that treat interference as noise or decode interference are addressed. 
\end{abstract}
\begin{IEEEkeywords}
Achievable Outage Level, Block Transmission, Gradual Data Arrival, Interference Channel, Tx-Tx Asynchrony.
\end{IEEEkeywords}
\section{Introduction}
The interference channel~(IC) is the basic building block in modelling ad hoc wireless networks of separate Tx-Rx pairs. The Shannon capacity region of interference channels has been unknown for more than thirty years. It is shown in \cite{et1} that the classical random coding scheme developed by Han and Kobayashi~\cite{HK} achieves within one bit of the capacity region of the two-user Gaussian~IC for all ranges of channel coefficients and signal-to-noise ratio values. 

One key assumption made in \cite{et1} and the references therein is that data is constantly available at the encoders, i.e., the transmitters are backlogged. This is in contrast with the reality of communication networks where the incoming bit streams at the transmitters are bursty in nature.  

There exist several papers that study a Gaussian IC with bursty data arrival.  The authors in~\cite{sim} address a cognitive interference channel where the data arrival is bursty at both the primary and secondary transmitters. The secondary transmitter regulates its average transmission power in order to maximize its throughput subject to two conditions, namely, the throughput of the primary user does not fall below a given threshold and the queues of both transmitters remain stable.\footnote{Stability of a queue is in the sense that the cumulative distribution function of the length of queue converges pointwise to a limiting distribution in the long run when the number of time slots grows to infinity.}  
The exact stability region of the two-user GIC was characterized in~\cite{pappas}. More recently, \cite{haenggi} introduced the notion of $\epsilon$-stability region in a static wireless network of multiple Tx-Rx pairs that are distributed spatially according to a Poisson point process and share a common frequency band under the random access protocol. 

The system model considered in this paper differs from those in \cite{sim}, \cite{pappas} and \cite{haenggi} in the following aspects: 
\begin{enumerate}
  \item References \cite{sim}, \cite{pappas} and \cite{haenggi} deal with the problem of interacting queues~\cite{Rao}. Each communicating pair is equipped with a feedback channel through which the receiver sends a request for retransmission to its affiliated transmitter in case a packet is lost. Consequently, the service rate at each queue becomes dependent on the contents of all queues. In the current paper, no such acknowledgements are made due to the absence of feedback. If a transmitted codeword is lost, it is lost forever. This decouples the dynamics of different queues, i.e., they no longer interact. 
  \item The adopted model for data arrival and the definition of ``time slot'' in the current paper are different from those in \cite{sim}, \cite{pappas} and \cite{haenggi}. In these works, a whole codeword arrives in a single time slot at each transmitter with a given probability.  Also, a codeword is transmitted in one time slot if the queue is nonempty. In the current paper,  each queue is likely to receive only a fixed number of bits during each time slot. These bits keep accumulating in the buffer until their number is large enough to represent a codeword. 
\end{enumerate} 

The rest of the paper is organized as follows. The system model, problem formulation and a statement of results are presented in Section~II. The block transmission scheme is introduced in Section~III. Section~IV studies sufficient conditions for successful communication. Section~V is devoted to computing the probability of outage under the block transmission scheme. Section VI explores the general bounds of Theorem~1 in the case of a Gaussian~IC. Finally, Section~VII concludes the paper.

\begin{figure}[t]
  \centering
  \includegraphics[scale=1.1] {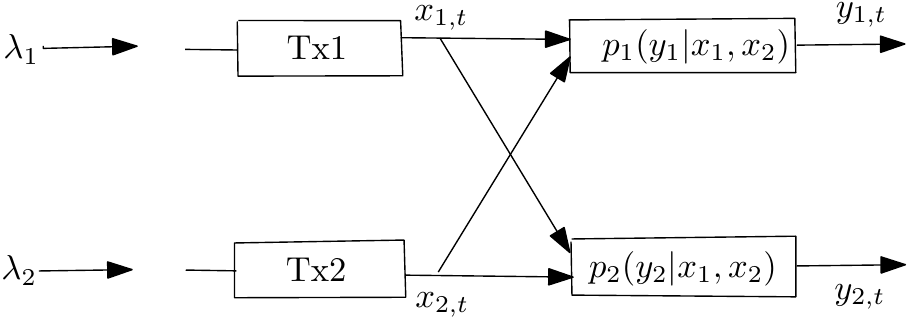}
  \caption{This figure shows a two-user discrete memoryless IC with gradual data arrival in the absence of feedback. In this paper, we study the scenario where the rates of data arrival at the transmitters are identical, i.e., $\lambda_1=\lambda_2=\lambda$.   }
  \label{fig1}
 \end{figure} 
 \section{System Model and Problem Formulation}
We consider a discrete memoryless interference channel of two separate~Tx-Rx pairs in the absence of feedback shown in Fig.\!~\ref{fig1}. We denote the time slots using the index $t=1,2,\cdots$ and the signals at Tx~$i$ and Rx~$i$ during time slot $t$ by $x_{i,t}\in \mathcal{X}_i$ and $y_{i,t}\in \mathcal{Y}_i$, respectively, where $\mathcal{X}_i$ and $\mathcal{Y}_i$ are the correspoding alphabets. Each transmitter is connected to one information source.  The information source at Tx~$i$ generates an i.i.d. Bernoulli process with parameter $\frac{1}{2}$, i.e., each generated bit is a $1$ with a probability of $\frac{1}{2}$ and a $0$ with a probability of $\frac{1}{2}$. The bits arrive gradually. During each time slot, a single bit arrives at Tx~$i$ with a probability of $\lambda_i$ or no bit arrives with a probability of $1-\lambda_i$.
 The bit streams at the two transmitters are independent processes. Each information source turns off after generating $n$ bits, i.e., the communication load per transmitter is $n$ bits. 
 The two data streams are asynchronous. The activation time slot for the data stream at Tx~$i$ is $ \lceil \boldsymbol{d}_in\rceil$ where $\boldsymbol{d}_1, \boldsymbol{d}_2$ are independent random variables uniformly distributed on the interval $[0,D]$ for given $D>0$. Throughout the paper, random variables appear in bold such as $\boldsymbol{x}$ with realization $x$.  
 
  Let $(b_{i,1},b_{i,2},\cdots, b_{i,n})$ be the sequence of $n$ bits generated by the source at Tx~$i$. An encoder at Tx~$i$ is a mapping $f_{i,n}: \{0,1\}^n\to \mathcal{X}_i^{T_{i,n}}$ that maps $(b_{i,1},b_{i,2},\cdots, b_{i,n})$ into a number of $T_{i,n}$ signals $x_{i,t}$ for $t\in \mathcal{T}_i$ where $\mathcal{T}_i$ is the activity period of length $T_{i,n}$ for Tx~$i$    given by 
  \begin{eqnarray}
  \label{scv_111}
\mathcal{T}_i=\{ \lceil d_in\rceil ,\lceil d_in\rceil+1,\cdots, \lceil d_in\rceil+T_{i,n}-1\}.
\end{eqnarray}
  Note that $x_{i,t}$ depends only on those bits that arrive at or before time slot $t$. Moreover, $T_{i,n}$ is a realization of a random variable, because its value depends on how early the bits $b_{i,1},b_{i,2},\cdots, b_{i,n}$ arrive at the encoder. We define the transmission rate at encoder $f_{i,n}$ by 
  \begin{eqnarray}
R_{i,n}:=\frac{n}{T_{i,n}}.
\end{eqnarray}
  A decoder at Rx~$i$ is a mapping $g_{i,n}:\mathcal{Y}_i^{T_{i,n}}\to \{0,1\}^n$ that maps the received signals $y_{i,t}$ for $t\in \mathcal{T}_i$ into a sequence of $n$ bits $(\hat{b}_{i,1},\hat{b}_{i,2},\cdots, \hat{b}_{i,n})$ where $ \hat{b}_{i,j}$ is the estimated value for $b_{i,j}$. Given realizations $d_1, d_2$ for $\boldsymbol{d}_1,\boldsymbol{d}_2$, the bit-error-rate at Rx~$i$ for the encoder-decoder pair $(f_{i,n},g_{i,n})$ is defined by 
  \begin{eqnarray}
  \label{boolk_123}
p^{(b)}_{i,n}(d_1,d_2):=\frac{1}{n}\sum_{j=1}^n\Pr\big(\hat{\boldsymbol{b}}_{i,j}\neq\boldsymbol{b}_{i,j}|\boldsymbol{d}_1=d_1,\boldsymbol{d}_2=d_2\big).
\end{eqnarray}

We say $(\epsilon_1,\epsilon_2)$ is a pair of \textit{achievable outage levels} if there exists a sequence of encoder-decoder pairs $(f_{i,n},g_{i,n})$ with transmission rates $\boldsymbol{R}_{i,n}$ and bit-error-rates $p^{(b)}_{i,n}(\boldsymbol{d}_1,\boldsymbol{d}_2)$ such that 
\begin{enumerate}[(i)]
  \item $\lim_{n\to\infty}\boldsymbol{R}_{i,n}= \lambda_i$ almost surely  for $i=1,2$ and
  \item $\Pr(\mathcal{O}_i) \le\epsilon_i $ for $i=1,2$ where $\mathcal{O}_i$ is the outage event that the bit-error-rate at Rx~$i$ does not converge to zero, i.e.,
\begin{eqnarray}
\label{outm}
\mathcal{O}_i=\big\{\limsup_{n\to\infty}p^{(b)}_{i,n}(\boldsymbol{d}_1,\boldsymbol{d}_2)>0\big\}.
\end{eqnarray}
\end{enumerate}
Let $\mathcal{E}$ be the set of all pairs of achievable outage levels.  The infimum of $\max\{\epsilon_1,\epsilon_2\}$ over all $(\epsilon_1,\epsilon_2)\in \mathcal{E}$ is denoted by $\epsilon(\lambda_1,\lambda_2)$ and referred to as the \textit{worst-case outage level}, i.e., 
\begin{eqnarray}
\epsilon(\lambda_1,\lambda_2)=\inf_{(\epsilon_1,\epsilon_2)\in \mathcal{E}} \max\{\epsilon_1,\epsilon_2\}
\end{eqnarray}
In this paper we focus on a scenario where the incoming data rates are identical, i.e., 
\begin{eqnarray}
\lambda_1=\lambda_2=\lambda.
\end{eqnarray}
 Our contribution is an upper bound (achievability result) on $\epsilon(\lambda,\lambda)$ which we denote by $\epsilon(\lambda)$ for short. 

Next, we present a statement of the main result. Let $\mathrm{Prob}(\mathcal{X}_i)$ be the set of probability distributions on the input alphabet $\mathcal{X}_i$. Fix $x^*_1\in \mathcal{X}_1$, $x^*_2\in \mathcal{X}_2$, $\pi_1\in \mathrm{Prob}(\mathcal{X}_1)$ and $\pi_2\in \mathrm{Prob}(\mathcal{X}_2)$. Let $\boldsymbol{x}_1$ and $\boldsymbol{x}_2$ be independent random variables with values in $\mathcal{X}_1$ and $\mathcal{X}_2$ whose distributions are $\pi_1$ and $\pi_2$, respectively. Let  $\boldsymbol{y}_1$ and $\boldsymbol{y}_2$ be the corresponding random variables at the outputs of the interference channel. For $i\in \{1,2\}$ denote $3-i$ by $i'$ and define
\begin{eqnarray}
\label{ino_1}
C_{i}^*:=I(\boldsymbol{x}_i;\boldsymbol{y}_i|\boldsymbol{x}_{i'}=x^*_{i'}),
\end{eqnarray}
\begin{eqnarray}
C_i:=I(\boldsymbol{x}_i;\boldsymbol{y}_i)
\end{eqnarray}
and
\begin{eqnarray}
\label{ciip}
C_{i,i'}:=I(\boldsymbol{x}_i;\boldsymbol{y}_i|\boldsymbol{x}_{i'})
\end{eqnarray}
Also, define
\begin{eqnarray}
\tilde{C}_{i}^*:=I(\boldsymbol{x}_{i'};\boldsymbol{y}_i|\boldsymbol{x}_{i}=x^*_{i})
\end{eqnarray}
and
\begin{eqnarray}
\label{ino_2}
\tilde{C}_i:=I(\boldsymbol{x}_{i'};\boldsymbol{y}_i).
\end{eqnarray}
A standard application of data processing inequality shows that if 
\begin{eqnarray}
\label{l_bar}
\lambda> \overline{\lambda}:=\min_{i=1,2}\max_{\pi_1,\pi_2}C_{i,i'},
\end{eqnarray}
then there is a user for which there exists no coding scheme that guarantees a vanishingly small bit error rate. The proof of this statement is provided in Appendix~A. 
\begin{thm}
\label{thm1}
 Assume $\lambda\le\overline{\lambda}$ where $\overline{\lambda}$ is given in (\ref{l_bar}). Define\footnote{The superscripts TIN and DI represent treating interference as noise and decoding interference, respectively. } 
 \begin{eqnarray}
 \label{l_tin}
 \lambda^{(TIN)}:=\min\{C_1,C_2\}
\end{eqnarray}
and 
\begin{eqnarray}
\label{l_di}
\lambda^{(DI)}:=\min\{C_{1,2}, \tilde{C}_1, C_{2,1}, \tilde{C}_2\}.
\end{eqnarray}
For $r> 1$ define 
 \begin{eqnarray}
\rho^{(TIN)}_i(r):= \frac{\lambda r-C_i}{C_i^*-C_i} 
\end{eqnarray}
and 
\begin{eqnarray}
\rho^{(DI)}_i(r):= \max\Big\{\frac{\lambda r-C_{i,i'}}{C_i^*-C_{i,i'}},\frac{\lambda r-\tilde{C}_{i}}{\tilde{C}_i^*-\tilde{C}_{i}}\Big\}.
\end{eqnarray}
\newpage
\begin{enumerate}
  \item If $\lambda< \lambda^{(TIN)}$ or $\lambda< \lambda^{(DI)}$, then $\epsilon(\lambda)=0$.
  \item If $\lambda\ge\lambda^{(TIN)}$, then 
  \begin{eqnarray}
  \epsilon(\lambda)\leq \epsilon^{(TIN)}(\lambda):=\frac{\kappa}{r_0}\max\Big\{\rho_1^{(TIN)}(r_0), \rho_2^{(TIN)}(r_0)\Big\},
\end{eqnarray}
where $r_0$ is the infimum of all $r>1$ that satisfy $\max\big\{\rho_1^{(TIN)}(r),\rho_2^{(TIN)}(r)\big\}<\min\{1,r-1\}$ and  
 \begin{eqnarray}
 \label{capp}
\kappa:=\left\{\begin{array}{cc}
    \frac{2}{\lambda D}(2-\frac{1}{\lambda D})  &  \lambda D\geq 1  \\
   2   &   \lambda D< 1
\end{array}\right..
\end{eqnarray}
 \item If $\lambda\ge\lambda^{(DI)}$, then 
  \begin{eqnarray}
  \epsilon(\lambda)\leq \epsilon^{(DI)}(\lambda):=\frac{\kappa}{r_0}\max\Big\{\rho_1^{(DI)}(r_0), \rho_2^{(DI)}(r_0)\Big\},
\end{eqnarray}
where $r_0$ is the infimum of all $r>1$ that satisfy  $\max\big\{\rho_1^{(DI)}(r),\rho_2^{(DI)}(r)\big\}<\min\{1,r-1\}$. 
\end{enumerate}
\end{thm}
\begin{proof}
The proof builds upon Sections III, IV and V. 
\end{proof}

\section{A Block Transmission Scheme}
To transmit its data, each transmitter employs a codebook consisting of $2^{ n\theta R_c }$ codewords of length~$n\theta$
 where $\theta>0$ is a parameter to be determined shortly and $R_{c}$ is the \textit{code~rate}.\footnote{To be accurate, we need to write that the codebook consists of $2^{\lfloor n\theta R_c \rfloor}$ codewords of length~$\lfloor n\theta\rfloor$. For notational simplicity, we have dropped the floor symbol $\lfloor\cdot\rfloor$.} 
 Upon transmission, a codeword is sent over the channel in~$n\theta$ consecutive time~slots. For simplicity of presentation, let us temporarily assume Tx~$i$ starts its activity at time slot $t=1$.  Let $\boldsymbol{s}_{i,l}$ be the number of bits entering the buffer of Tx~$i$ during time slot $l$ and $\boldsymbol{\tau}_{i,1}$ be the smallest index $t\geq 1$ such that $\sum_{l=1}^t\boldsymbol{s}_{i,l}\ge n\theta R_c $, i.e. 
 \begin{equation}
\label{ }
\boldsymbol{\tau}_{i,1}:=\min\Big\{t\geq 1: \sum_{l=1}^t\boldsymbol{s}_{i,l}\ge  n\theta R_c \Big\}.
\end{equation}
At time slot $t=\boldsymbol{\tau}_{i,1}+1$, a number of $n\theta R_{c}$ bits in the buffer of Tx~$i$ are represented by a codeword which is transmitted over the channel during time slots of indices $\boldsymbol{\tau}_{i,1}+1,\cdots, \boldsymbol{\tau}_{i,1}+n\theta$. Let the total number of transmitted codewords per user be $N$. Since each transmitter only sends a total number of $n$ bits, we need $N\times n\theta R_c=n$, i.e., 
\begin{eqnarray}
\theta=\frac{1}{NR_c}.
\end{eqnarray}
We adopt the notation $\theta=\theta_{N}$ hereafter to make the dependence of $\theta$ on $N$ explicit.   For every~$2\le j\leq N$, define
\begin{equation}
\label{buffer1}
\boldsymbol{\tau}_{i,j}:=\min\Big\{t\geq \boldsymbol{\tau}_{i,j-1}+ n\theta_{N}: \sum_{l=1}^t\boldsymbol{s}_{i,l}\ge j n\theta_{N} R_{c}\Big\}.
\end{equation}
 At time slot $\boldsymbol{\tau}_{i,j}$, the $j^{th}$ group of $n\theta_{N} R_{c}$ bits in the buffer of Tx~$i$ are represented by a codeword which is transmitted over the channel during time slots of indices $\boldsymbol{\tau}_{i,j}+1,\cdots, \boldsymbol{\tau}_{i,j}+ n\theta_{N}$. 
We assume $\boldsymbol{x}_{i,t}=x_i^*$ for every time slot $t$ when Tx~$i$ is idle, i.e., it does not transmit a symbol of a codeword. Here, $x_i^*\in \mathcal{X}_i$ is a fixed member of $\mathcal{X}_i$.  

Define  the random time
\begin{eqnarray}
\label{inu_1}
\boldsymbol{\xi}_{i,j}:=\min\Big\{t\geq 1: \sum_{l=1}^t\boldsymbol{s}_{i,l}\ge j n\theta_{N} R_{c}\Big\}.
\end{eqnarray}
Then $\boldsymbol{\tau}_{i,j}$  in (\ref{buffer1}) satisfies the recursion
\begin{eqnarray}
\label{tweet1}
\begin{array}{c}
      \boldsymbol{\tau}_{i,1}=\boldsymbol{\xi}_{i,1}  \\
         \boldsymbol{\tau}_{i,j}=\max\big\{\boldsymbol{\tau}_{i,j-1}+n\theta_{N}, \boldsymbol{\xi}_{i,j}\big\},\,\,\,j=2,3,\cdots, N
\end{array}.
\end{eqnarray}
We know\footnote{Recall that $\boldsymbol{s}_{i,l}\in\{0,1\}$ is the \textit{number} of bits that enter the buffer of Tx~$i$ during time slot $l$.  } $\boldsymbol{s}_{i,l}$ for $l=1,2,\cdots$ is an i.i.d. Bernoulli process with $\Pr(\boldsymbol{s}_{i,l}=1)=\lambda$ and $\Pr(\boldsymbol{s}_{i,l}=0)=1-\lambda$. As such, $\boldsymbol{\xi}_{i,j}$ is a negative binomial random variable with parameters $ jn\theta_{N} R_c$ and $\lambda$, i.e.,   
\begin{eqnarray}
\label{lln1}
\boldsymbol{\xi}_{i,j}\sim \mathrm{NB}( jn\theta_{N} R_c,\lambda).
\end{eqnarray}
Interpreting $\boldsymbol{\xi}_{i,j}$ as a sum of $ jn\theta_{N} R_c$ independent geometric random variables with parameter $\lambda$, one can invoke the strong law of large numbers to conclude that $\lim_{n\to\infty}\frac{\boldsymbol{\xi}_{i,j}}{jn\theta_{N} R_c}=\frac{1}{\lambda}$. Then
\begin{eqnarray}
\label{powl44}
\lim_{n\to\infty}\frac{\boldsymbol{\xi}_{i,j}}{n\theta_{N}}=\frac{j R_c}{\lambda}=j r,
\end{eqnarray} 
where we have defined the \textit{normalized code rate} $ r$  by   
\begin{eqnarray}
 r:=\frac{R_c}{\lambda}.
\end{eqnarray}
By (\ref{tweet1}) and (\ref{powl44}) and using induction on $j$, we find that the limit 
\begin{eqnarray}
\label{powl55}
\overline{\tau}_{j}:=\lim_{n\to\infty}\frac{\boldsymbol{\tau}_{i,j}}{n\theta_{N}}
\end{eqnarray}
exists for every $j\geq 1$ and satisfies the recursion 
\begin{eqnarray}
\label{tweet3}
\begin{array}{c}
      \overline{\tau}_{1}= r  \\
         \overline{\tau}_{j}=\max\big\{\overline{\tau}_{j-1}+1, j r\big\},\,\,\,j=2,3,\cdots,N
\end{array}.
\end{eqnarray}
Using induction one more time, we get 
\begin{eqnarray}
\label{shei_1}
\overline{\tau}_{j}=\left\{\begin{array}{cc}
    j r  &  r> 1   \\
    r+j-1   &   0< r\le1
\end{array}\right.,
\end{eqnarray}
for every $1\le j\le N$. It is convenient to define the scaled time-axis $\overline{t}=\frac{t}{n\theta_{N}}$ which we refer to as the $\overline{t}$-axis.  For $i=1,2$ and $1\le j\le N$ define the activity intervals \begin{eqnarray}
\label{act_int_as}
\mathcal{I}_{i,j}:=\big(d_i/\theta_{N}+\overline{\tau}_{j},d_i/\theta_{N}+\overline{\tau}_{j}+1\big).
\end{eqnarray}
The following proposition summarizes the observations made in above: 
\begin{proposition}
\label{prop_2}
For every $1\le j\le N$, the $j^{th}$ codeword of Tx~$i$ lies almost surely over $\mathcal{I}_{i,j}$ on the $\overline{t}$-axis.
\end{proposition} 
\textit{Remark~1}- We see that if $ r\le1$, then $\overline{\tau}_{j+1}-\overline{\tau}_{j}=1$ for every $1\le j\le N-1$ and hence, there is no gap between consecutively transmitted codewords along the $\overline{t}$-axis. If $ r>1$, there always exists a nonzero gap between every two codewords, i.e.,  the signals sent by each transmitter look like intermittent~bursts.~$\Diamond$ 

\textit{Remark~2}- In practice both transmitters insert a preamble sequence at the beginning of each codeword. The receivers then apply standard detection techniques such as sequential joint-typicality detection~\cite{Chandar11} in order to find the exact arrival time of the codewords. The length of preamble sequences is of the order of $o(n)$, say logarithmic in $n$, and hence, the users do not incur any loss in spectral efficiency as $n$ tends to infinity.~$\Diamond$  

We end this section by finding an expression for the average transmission rate per user. Tx~$i$ sends a total number of $n$ bits over a period of $\boldsymbol{T}_{i,n}=\boldsymbol{\tau}_{i,N}+n\theta_{N}$ time slots. Then the average transmission rate for Tx~$i$ is given by the random variable 
\begin{eqnarray}
\boldsymbol{R}_{i,n}=\frac{n}{\boldsymbol{T}_{i,n}}=\frac{n}{\boldsymbol{\tau}_{i,N}+n\theta_{N}}.
\end{eqnarray}
As $n$ grows to infinity, $\boldsymbol{R}_{i,n}$ converges to the constant 
\begin{eqnarray}
R(N,r):=\lim_{n\to\infty}\boldsymbol{R}_{i,n}=\lim_{n\to\infty}\frac{n}{n\theta_N(\frac{\boldsymbol{\tau}_{i,N}}{n\theta_N}+1)}=\frac{Nr}{\overline{\tau}_{N}+1}\lambda.
\end{eqnarray} 
 Replacing $\overline{\tau}_{N}$ by its value given in~(\ref{shei_1}), we get 
\begin{eqnarray}
\label{yar_11}
R(N,r)=\left\{\begin{array}{cc}
    \frac{Nr}{Nr+1}\lambda  & r>1   \\
   \frac{Nr}{N+r}\lambda   &   0<r\le1
\end{array}\right..
\end{eqnarray} 
Condition~(i) in the definition of a pair of achievable outage levels requires the average transmission rate be $\lambda$. Considering that $r$ can not grow arbitrarily large due to the limited capacity of the interference channel,  there are two scenarios that guarantee $R(N,r)$ approaches $\lambda$:
\begin{enumerate}
  \item Let $r\geq1$. Then $\lim_{N\to\infty}R(N,r)=\lambda$ regardless of $r$.
  \item Let $r<1$. Then $\lim_{r\to1^-}\lim_{N\to\infty}R(N,r)=\lambda$.
\end{enumerate}
Interestingly, the first scenario with $r>1$ achieves smaller outage probabilities compared to the second scenario. Appendix~B presents the outage probabilities achieved under the second scenario. The reader is encouraged to compare the observations made in Appendix B with the results we derive in Section~IV. Our next goal is to compute the limit $\inf_{r> 1 }\lim_{N\to\infty}\max\{\Pr(\mathcal{O}_1),\Pr(\mathcal{O}_2)\}$ under the block transmission scheme.

\section{Conditions for successful communication when $r>1$ }
   In this section we obtain sufficient conditions such that the transmitted codewords are decoded successfully at their respective receivers when $r> 1$. We assume Tx~$i$ employs random codes generated according to $\pi_i\in \mathrm{Prob}(\mathcal{X}_i)$. Two decoding strategies are examined, namely, treating interference as noise and decoding interference. 
   

 Let
 \begin{eqnarray}
\label{diff_delay}
\delta:=|d_1-d_2|
\end{eqnarray}
and
 \begin{eqnarray}
 \label{lola}
\rho_i(r):=\left\{\begin{array}{cc}
  \frac{\lambda r-C_i}{C_i^*-C_i}  &  \textrm{Rx~$i$ treats interference as noise}  \\
  \max\Big\{\frac{\lambda r-C_{i,i'}}{C_i^*-C_{i,i'}},\frac{\lambda r-\tilde{C}_{i}}{\tilde{C}_i^*-\tilde{C}_{i}}\Big\}  & \textrm{Rx~$i$ decodes interference}  
\end{array}\right..
\end{eqnarray}
Depending on the decoder, we denote $\rho_i(r)$ by $\rho^{(TIN)}_i(r)$ or $\rho^{(DI)}_i(r)$.
Define the \textit{admissible intervals} $\mathcal{A}_{i,j}$ by  
  \begin{eqnarray}
  \label{khia_11}
\mathcal{A}_{i,j}:=\big((j-1) r+\rho_i(r),j r-\rho_i(r)\big),\,\,\,\,\,\,\,\,\,\,\,1\leq j\leq N-1
\end{eqnarray}
and 
 \begin{eqnarray}
\mathcal{A}_{i,N}:=\big((N-1) r+\rho_i(r),\infty\big).
\end{eqnarray}

\begin{proposition}
\label{prop_3}
  All $N$ transmitted codewords intended for Rx~$i$ are decoded successfully at this receiver if either the two conditions\footnote{For an interval $I=(a,b)$ and $c>0$, we denote $(ca,cb)$ by $cI$.}
  \begin{eqnarray}
  \label{dam_11}
 \rho_i(r)<\min\{1,r-1\},\,\,\,\,\,\,\,\delta\in \bigcup_{j=1}^{N-1}\theta_N\mathcal{A}_{i,j}
\end{eqnarray}
hold or the two conditions
 \begin{eqnarray}
\rho_i(r)<1,\,\,\,\,\,\,\,\delta\in \theta_N\mathcal{A}_{i,N}
\end{eqnarray}
hold.
\end{proposition}
\begin{proof}
The criterion for successful communication is $R_c<I$ where $R_c=r\lambda$ is the code rate and $I$ is the ``effective'' mutual information from a transmitter to a receiver. Note that $I$ depends on the relative positions of activity intervals for the two transmitters given in (\ref{act_int_as}) which in turn depend on $R_c$ itself. The details are provided in Appendix~C. 
\end{proof}  
 If there exists $r>1$ such that $\rho_i(r)<0$, then the conditions in Proposition~\ref{prop_3} hold regardless of the value of $\delta\ge0$.  In this case, $\Pr(\mathcal{O}_i)=0$. Throughout the rest of the paper, we assume regardless of $r> 1$, either $\rho_1(r)\geq 0$ or $\rho_2(r)\ge0$. This is easily seen to be equivalent to 
\begin{eqnarray}
\lambda\ge \lambda^{(TIN)}:=\min\{C_1,C_2\}
\end{eqnarray}
for receivers that treat interference as noise and to 
\begin{eqnarray}
\lambda\geq\lambda^{(DI)}:=\min\{C_{1,2}, \tilde{C}_1, C_{2,1}, \tilde{C}_2\}
\end{eqnarray}
for receivers that decode interference. The first case in Theorem~1 is already verified. During Section~V, we will verify the second and third cases.  

\textit{Remark 3}- It may happen that $C_i^*=C_{i,i'}$. For example, additive channels satisfy this condition.  Then the term $\frac{\lambda r-C_{i,i'}}{C_i^*-C_{i,i'}}$ in (\ref{lola}) does not make sense. In this case, we need to redefine $\rho^{(DI)}_i(r)=\frac{\lambda r-\tilde{C}_{i}}{\tilde{C}_i^*-\tilde{C}_{i}}$ and replace the first inequality in (\ref{dam_11}) by the two inequalities $ \rho^{(DI)}_i(r)=\frac{\lambda r-\tilde{C}_{i}}{\tilde{C}_i^*-\tilde{C}_{i}}<\min\{1,r-1\}$ and $r<\frac{C_i^*}{\lambda}$.~$\Diamond$ 
\section{The outage probability under the block transmission scheme}
   In practice the transmitters do not know each other's activation times $\lceil nd_1\rceil$ and $\lceil nd_2\rceil$ and hence, the value of~$\delta$~in~(\ref{diff_delay}) is unknown to both transmitters. In this section, we treat $\delta$ as a random variable  and compute the probability of communication failure (outage) under the block transmission scheme. We know $\boldsymbol{d}_1$ and $\boldsymbol{d}_2$ are independent uniform random variables  over the interval $[0,D]$ for some $D>0$. Then it is easy to see that the cumulative distribution function~(CDF) of the Tx-Tx asynchrony $\boldsymbol{\delta}=|\boldsymbol{d}_2-\boldsymbol{d}_1|$ is given by    
\begin{eqnarray}
\label{cdf_11}
F_{\boldsymbol{\delta}}(\delta)=\left\{\begin{array}{cc}
    0  &   \delta<0 \\
   \frac{ \delta}{D}\big(2-\frac{\delta}{D}\big)  &   0\leq \delta<D\\
    1 & \delta\geq D
\end{array}\right..
\end{eqnarray}
Let $i\in\{1,2\}$ and $\rho_i(r)\ge0$. We will use (\ref{cdf_11}) in order to compute an upper bound on $\Pr(\mathcal{O}_i)$ where $\mathcal{O}_i$ is the outage event that at least one of the $2N$ transmitted packets is not decoded successfully at Rx~$i$.  By Proposition~\ref{prop_3}, 
\begin{eqnarray}
 \Big\{\boldsymbol{\delta} \in \bigcup_{j=1}^{N-1}\theta_N\mathcal{A}_{i,j}\,\,\,\textrm{and}\,\,\,\rho_i(r)<\min\{1,r-1\}\Big\}\bigcup  \Big\{\boldsymbol{\delta} \in \theta_N\mathcal{A}_{i,N}\,\,\,\textrm{and}\,\,\,\rho_i(r)<1\Big\}\subseteq \mathcal{O}^c_{i},
\end{eqnarray}
where $\mathcal{O}_i^c$ is the complement of $\mathcal{O}_i$. Then we obtain
\begin{eqnarray}
\Pr(\mathcal{O}_i)\leq p^{(ub)}_{outage}(i):=1-\Pr\Big(\boldsymbol{\delta} \in \bigcup_{j=1}^{N-1}\theta_N\mathcal{A}_{i,j}\Big)\chi_{1}(i)-\Pr(\boldsymbol{\delta} \in \theta_N\mathcal{A}_{i,N})\chi_{2}(i),
\end{eqnarray}
where the indicator variables $\chi_1(i)$ and $\chi_2(i)$ are given by
\begin{eqnarray}
\chi_{1}(i)=\left\{\begin{array}{cc}
   1   &  \rho_i(r)<\min\{1,r-1\}  \\
    0  &   \mathrm{otherwise}
\end{array}\right.,\hskip1cm\chi_{2}(i)=\left\{\begin{array}{cc}
   1   &  \rho_i(r)<1  \\
    0  &   \mathrm{otherwise}
\end{array}\right.
\end{eqnarray}
We are interested in computing $\lim_{N\to\infty}p^{(ub)}_{outage}(i)$. As $N$ increases, the number of intervals in the union $ \bigcup_{j=1}^{N-1}\theta_N\mathcal{A}_{i,j}$ increases, however, the individual intervals $\theta_N \mathcal{A}_{i,j}$ shrink.  The next proposition presents closed-form expressions for $p^{(ub)}_{outage}(i)$ in terms of $N$.
\newpage
\begin{proposition}
\label{prop_7}
Assume $\rho_i(r)>0$. Define
\begin{eqnarray}
\label{alpha_11}
\alpha:=\lambda D
\end{eqnarray} 
\begin{eqnarray}
\label{beta_11}
\beta_i:=\rho_i(r)/r
\end{eqnarray}
and let
\begin{eqnarray}
\label{back_11}
m:=\lceil N\alpha-\beta_{i}\rceil.
\end{eqnarray} 
\begin{itemize}
  \item If $m=0$ or $1\leq m\leq N-1$ is such that $m\le N\alpha+\beta_{i}$, then 
  \begin{eqnarray}
  \label{back_22}
p^{(ub)}_{outage}(i)=1-(1-2\beta_{i})\Big(2- \frac{m}{N\alpha}\Big)\frac{m}{N\alpha}\chi_{1}(i).
\end{eqnarray}
  \item If $1\leq m\leq N-1$ is such that $m>N\alpha+\beta_{i}$, then 
  \begin{eqnarray}
  \label{pris_111}
p^{(ub)}_{outage}(i)=1-\bigg((1-2\beta_{i})\Big(2-\frac{m-1}{N\alpha}\Big)\frac{m-1}{N\alpha}-\Big(1-\frac{m-1+\beta_{i}}{N\alpha}\Big)^2\bigg)\chi_{1}(i).
\end{eqnarray}
  \item If $m\geq N$, then 
   \begin{eqnarray}
   \label{jame_11}
p^{(ub)}_{outage}(i)=1-(1-2\beta_{i})\Big(2-\frac{N-1}{N\alpha}\Big)\frac{N-1}{N\alpha}\chi_{1}(i)-\Big(1-\frac{N-1+\beta_{i}}{N\alpha}\Big)^2\chi_{2}(i).
\end{eqnarray}
\end{itemize}
\end{proposition}
\begin{proof}
See Appendix~D.
\end{proof}
\textit{Remark 4}- If $\chi_{1}(i)=1$, it is guaranteed that $0<\beta_i<\frac{1}{2}$. In fact, $\beta_i=\frac{\rho_i(r)}{r}\stackrel{(a)}{<} \frac{\rho_i(r)}{1+\rho_i(r)}\stackrel{(b)}{<}\frac{1}{2}$ where $(a)$ and $(b)$ are due to $\rho_i(r)<r-1$ and $\rho_i(r)<1$, respectively.~$\Diamond$ 


Before we compute $\lim_{N\to\infty}p^{(ub)}_{outage}(i)$, we continue with a few examples in order to describe the behaviour of $p^{(ub)}_{outage}(i)$ as a function of $N$. The main observation is that there exists a critical value for $\alpha$, denoted by $\alpha^*_i$, such that the behaviour of $p^{(ub)}_{outage}(i)$ as a function of $N$ undergoes a phase transition depending on whether $\alpha< \alpha_i^*$ or $\alpha\ge\alpha_i^*$. If $\alpha\ge\alpha^*_i$, then $p^{(ub)}_{outage}(i)$ becomes a nondecreasing function of~$N$. If $\alpha<\alpha_i^*$, then $p^{(ub)}_{outage}(i)$ becomes an oscillatory function of $N$ that converges to its absolute minimum value as $N$ grows large. Therefore, if $\alpha$ is sufficiently small, not only does the average transmission rate approach its largest value $\lambda$ by increasing $N$, but also the probability of outage decreases to its smallest value. 

\begin{figure}[t]
\centering
\subfigure[]{
\includegraphics[scale=0.465]{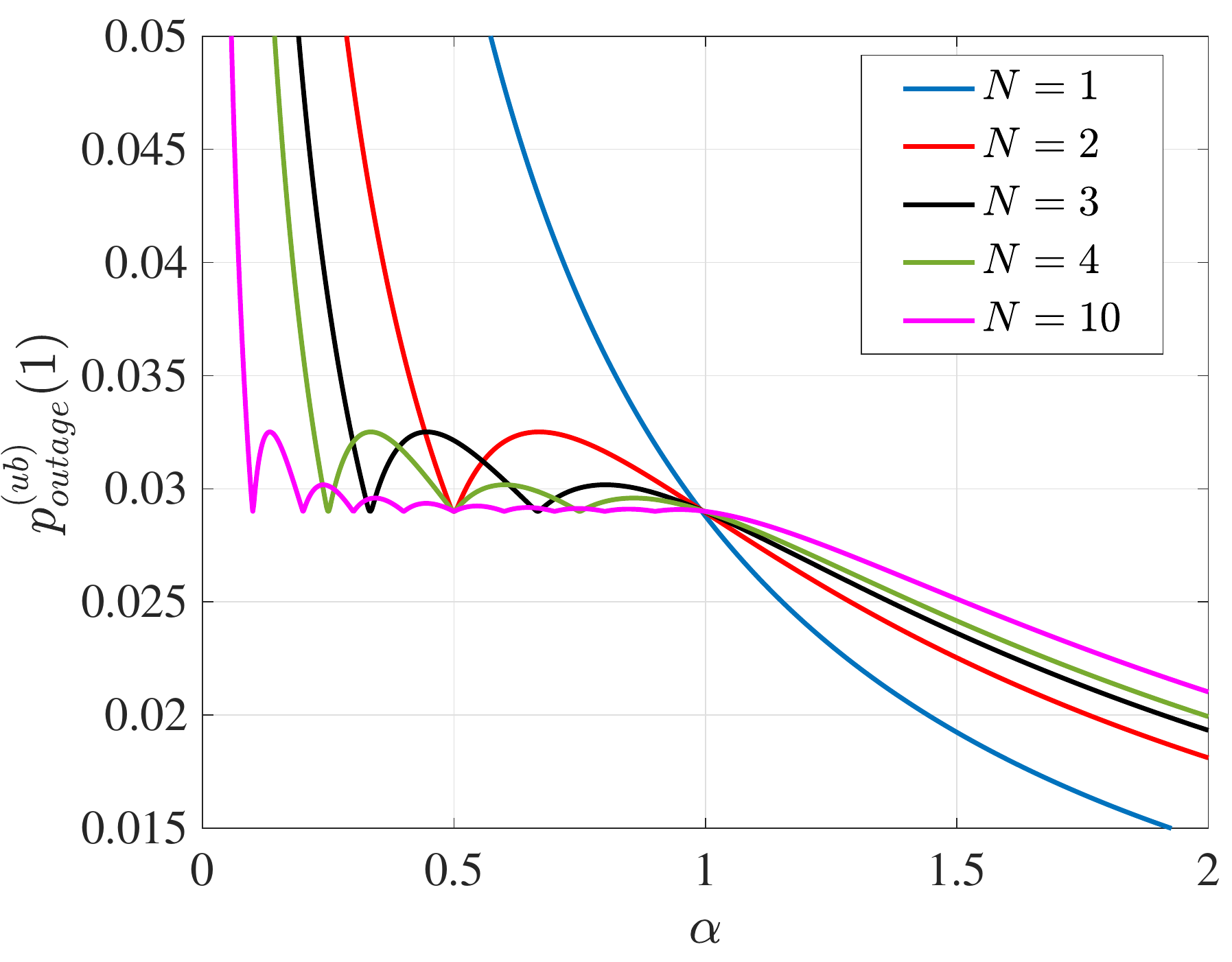}
\label{prs1_3}
}
\subfigure[]{
\includegraphics[scale=0.47]{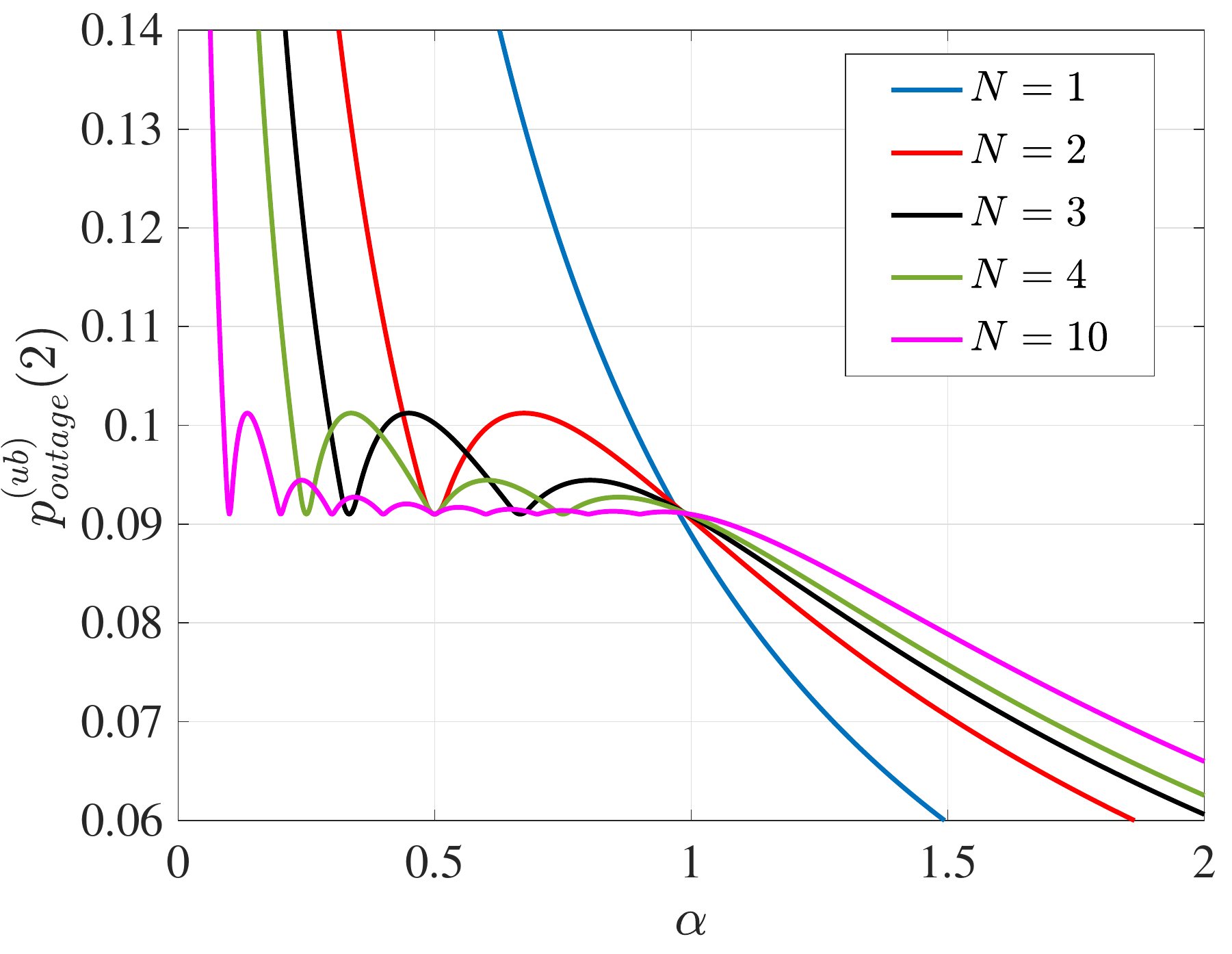}
\label{prs1_4}
}
\label{fig:subfigureExample}
\caption[Optional caption for list of figures]{ Let $\mathcal{X}_1=\mathcal{X}_2=\{0,1\}$, $\mathcal{Y}_1=\mathcal{Y}_2=\{0,1,2,3,4\}$ and the probability kernels $p_1(c|a,b)$ and $p_2(c|a,b)$ be identical given by the matrix in~(\ref{meho_00}). The system parameters are $\lambda=0.1$, $r=1.1$ and $\pi_1=\pi_2=\mathrm{Ber}(0.2)$. The receivers treat interference as noise. Panels (a) and (b) present plots of $p^{(ub)}_{outage}(1)$ and $p^{(ub)}_{outage}(2)$, respectively, as provided by Proposition~\ref{prop_7} in terms of $\alpha$ for different values of $N$. There is a phase transition in the behaviour of $p^{(ub)}_{outage}(i)$ as a function of $N$ around a critical value $\alpha_i^*\approx 1$ for $\alpha$. }
\label{prs1_11}
\end{figure}

\newpage
\textit{Example 1}- Consider a scenario where $\mathcal{X}_1=\mathcal{X}_2=\mathcal{X}=\{0,1\}$ and $\mathcal{Y}_1=\mathcal{Y}_2=\mathcal{Y}=\{0,1,2,3,4\}$. The marginal probability kernels $p_1(c|a,b)$ and $p_2(c|a,b)$ are identical for all $a,b\in \mathcal{X}$ and $c\in \mathcal{Y}$ described by the matrix\footnote{This transition matrix was generated randomly.}
\begin{eqnarray}
\label{meho_00}
\bordermatrix{(a,b)\textbackslash c&0&1&2&3&4\cr
                (0,0)&0.3266  &  0.1314  &  0.1674  &  0.3588  &  0.0158  \cr
                (0,1)& 0.3148  &  0.0612 &   0.2158  &  0.1898  &  0.2184 \cr
                (1,0)& 0.1905  &  0.3272 &   0.4279  &  0.0102  &  0.0442\cr
                (1,1)& 0.4091  &  0.2734  &  0.0970  &  0.1693  &  0.0512 }.
\end{eqnarray}
Assume the receivers treat interference as noise. It is easy to see that $\lambda^{(TIN)}$ defined in (\ref{l_tin}) is approximately equal to~$0.08$. Let $\lambda=0.1>\lambda^{(TIN)}$, $r=1.1$ and $\pi_1=\pi_2=\mathrm{Ber}(0.2)$. Then $\rho_1=0.016$, $\rho_2=0.0501$ and the condition $\rho_i<\min\{1,r-1\}$ in Proposition~\ref{prop_3} is met for both $i=1,2$. Fig.~\ref{prs1_11} presents plots of $p^{(ub)}_{outage}(i)$ for $i=1,2$ in terms of $\alpha=\lambda D$ for several values of $N$. There is a phase transition in the behaviour of $p^{(ub)}_{outage}(i)$ around a critical value $\alpha_i^*\approx 1$ for $\alpha$. If $\alpha\ge\alpha_i^*$, then $p^{(ub)}_{outage}(i)$ becomes increasing in terms of~$N$. If $\alpha<\alpha_i^*$, then $p^{(ub)}_{outage}(i)$ becomes an oscillatory function of $N$ that converges to its absolute minimum value as $N$ grows large.~$\Diamond$


  \newpage
\textit{Example 2}-  In this example we simplify the results of Proposition~\ref{prop_7} for  $N=1$ and $N=2$. Assume $0\le \rho_i<\min\{1,r-1\}$ so that $\chi_{1}(i)=\chi_{2}(i)=1$. Then
  \begin{eqnarray}
  \label{yasi_111}
p^{(ub)}_{outage}(i)\big|_{N=1}=\left\{\begin{array}{cc}
    1  & \beta_i\ge\alpha   \\
    \frac{\beta_i}{\alpha}(2-\frac{\beta_i}{\alpha})  &  \beta_i<\alpha
\end{array}\right..
\end{eqnarray}
and
  \begin{eqnarray}
  \label{hana_111}
  p^{(ub)}_{outage}(i)\big|_{N=2}=\left\{\begin{array}{cc}
    1  & \beta_i\ge2\alpha   \\
  1-(1-\frac{\beta_i}{2\alpha})^2  &   \beta_i<2\alpha\leq 1-\beta_i \\
     1-(1-2\beta_i)(2-\frac{1}{2\alpha})\frac{1}{2\alpha} & 1-\beta_i<2\alpha\leq 1+\beta_i\\
    1-(1-2\beta_i)(2-\frac{1}{2\alpha})\frac{1}{2\alpha}-(1-\frac{1}{2\alpha}(1+\beta_i))^2 & 1+\beta_i<2\alpha
\end{array}\right..
\end{eqnarray}
Assume $\alpha>\frac{1+\beta_i}{2}>\beta_i$. Simple algebra shows
\begin{eqnarray}
p^{(ub)}_{outage}(i)\big|_{N=2}-p^{(ub)}_{outage}(i)\big|_{N=1}=\frac{\beta_i}{\alpha^2}\Big(\frac{3\beta_i}{4}+\alpha-1\Big).
\end{eqnarray}
Then $N=2$ results in a smaller probability of outage compared to $N=1$ when $0\le\beta_i<\frac{2}{5}$ and $\frac{1+\beta_i}{2}<\alpha<1-\frac{3\beta_i}{4}$.~$\Diamond$
 

We now proceed to compute $\lim_{N\to\infty}p^{(ub)}_{outage}(i)$. This limit is particularly important  due to the fact that the average transmission rate $R(N,r)$ in (\ref{yar_11}) approaches its largest value $\lambda$ as $N$ grows to infinity. Recall that one of the conditions in the definition of an achievable outage level is that the average transmission rate be equal to $\lambda$. Inspecting the expressions for $p^{(ub)}_{outage}(i)$ given in Proposition~\ref{prop_7}, it is easy to see that  
  \begin{eqnarray}
  \label{out_inf}
\lim_{N\to\infty}p^{(ub)}_{outage}(i)=\left\{\begin{array}{cc}
 1-\frac{1}{\alpha}(2-\frac{1}{\alpha})(1-2\beta_i)\chi_{1}(i)-(1-\frac{1}{\alpha})^2\chi_{2}(i) &   \alpha\ge 1  \\
    1-(1-2\beta_i)\chi_{1}(i)  &  \alpha<1  
\end{array}\right.
\end{eqnarray}
This upper bound is the smallest when\footnote{In fact, $\chi_1(i)=1$ implies $\chi_2(i)=1$.} $\chi_{1}(i)=\chi_{2}(i)=1$. In this case, 
\begin{eqnarray}
\label{bet_1}
\lim_{N\to\infty}p^{(ub)}_{outage}(i)&=&\left\{\begin{array}{cc}
 \frac{2}{\alpha}(2-\frac{1}{\alpha})\beta_i &   \alpha\ge 1  \\
   2\beta_i &  \alpha<1  
\end{array}\right.\notag\\
&=&\kappa \beta_i,
\end{eqnarray}
where $\kappa$ is defined in (\ref{capp}). Hence, the worst-case probability of outage can be made as small as    
\begin{eqnarray}
\label{mod}
\inf_{r> 1}\lim_{N\to\infty}\max_{i=1,2} p^{(ub)}_{outage}(i)=\kappa\inf_{r>1: \chi_{1}(1)=\chi_{1}(2)=1 } \max\{\beta_1,\beta_2\},
\end{eqnarray}
Looking at the expressions for $\rho_i(r)$ in (\ref{lola}), we see that $\beta_i=\beta_i(r)=\rho_i(r)/r$ is increasing in terms of~$r$. Therefore, $\max\{\beta_1,\beta_2\}$ is also increasing in terms of $r$ and (\ref{mod}) can be written as 
\begin{eqnarray}
\inf_{r> 1}\lim_{N\to\infty}\max_{i=1,2} p^{(ub)}_{outage}(i)=\kappa\max\{\beta_1(r_{0}),\beta_2(r_{0})\},
\end{eqnarray}
where
\begin{eqnarray}
r_{0}&:=&\inf\{r>1: \chi_1(1)=\chi_1(2)=1\}\notag\\
&=&\inf\big\{r> 1: \max\{\rho_1(r), \rho_2(r)\}<\min\{1,r-1\}\big\}.
\end{eqnarray}
This concludes the proof of Theorem~\ref{thm1}. 
\section{Exploring the bounds in Theorem~1: The Gaussian Interference Channel}
Solving the inequality $\rho_i(r)<\min\{1,r-1\}$ involves one or two inequalities of the form $\frac{\lambda r-b}{a-b}<\min\{1,r-1\}$ for some $a,b>0$ when the receivers treat interference as noise or decode interference, respectively.  We will need the next lemma in order to interpret the results: 
   \begin{lem}
   \label{lem1}
   Let $a>b>0$. The solution to the inequality 
   \begin{eqnarray}
\frac{\lambda r-b}{a-b}<\min\{1,r-1\}
\end{eqnarray}
in the variable $r> 1$ is 
\begin{enumerate}[(i)]
  \item $1< r<\frac{a}{\lambda}$ when $\lambda<\min\{b,\frac{a}{2}\}$.
  \item $1< r<\frac{a-2b}{a-b-\lambda}$ when $\frac{a}{2}\leq \lambda<b$. Moreover, $1<\frac{a-2b}{a-b-\lambda}\le2$.
  \item $\frac{a-2b}{a-b-\lambda}<r<\frac{a}{\lambda}$ when $b\leq \lambda<\frac{a}{2}$. Moreover, $1\le\frac{a-2b}{a-b-\lambda}<2$.
  \item the empty set when $\lambda\geq \max\{b,\frac{a}{2}\}$. 
\end{enumerate}
   \end{lem}
   The elementary proof is omitted.
\newpage
  As an application of Lemma~1,  let us consider the Gaussian interference channel defined by  
   \begin{eqnarray}
   \label{gic}
\boldsymbol{y}_{i,t}=\boldsymbol{x}_{i,t}+\sqrt{c_i}\,\boldsymbol{x}_{i',t}+\boldsymbol{z}_{i,t}, \,\,\,\,\,i=1,2,
\end{eqnarray}
where $\boldsymbol{z}_{i,t}$ is a $N(0,1)$ random variable and the input of Tx~$i$ is subject to a maximum average transmission power of $P_i$. The numbers $\sqrt{c_1}$ and $\sqrt{c_2}$ are the crossover channel coefficients.  It is easy to see that $\overline{\lambda}$ in~(\ref{l_bar}) is given by 
\begin{eqnarray}
\overline{\lambda}=\min\{\mathsf{C}(P_1),\mathsf{C}(P_2)\},
\end{eqnarray}
where $\mathsf{C}(x)=\frac{1}{2}\log(1+x)$. For tractability reasons and in order to gain some insight, let us restrict our attention to point-to-point Gaussian codes. This amount to $\boldsymbol{x}_i$ in equations (\ref{ino_1}) to (\ref{ino_2}) being an $N(0,P_i)$ random variable for $i=1,2$. We get 
\begin{eqnarray}
C_i^*=C_{i,i'}=\mathsf{C}(P_i),\,\,\,\,\,\,\,C_i=\mathsf{C}\Big(\frac{P_i}{1+c_iP_{i'}}\Big)
\end{eqnarray}
and
\begin{eqnarray}
\tilde{C}_i^*=\mathsf{C}(c_i P_{i'}),\,\,\,\,\,\tilde{C}_i=\mathsf{C}\Big(\frac{c_i P_{i'}}{1+P_i}\Big).
\end{eqnarray}

 Let us address receivers that treat interference as noise and decode (and cancel) interference separately. 
   \begin{enumerate}
  \item Assume the receivers treat interference as noise. Then $\lambda^{(TIN)}$ in (\ref{l_tin}) is given by 
  \begin{eqnarray}
  \lambda^{(TIN)}=\min\Big\{ \mathsf{C}\Big(\frac{P_1}{1+c_1P_{2}}\Big), \mathsf{C}\Big(\frac{P_2}{1+c_2P_{1}}\Big) \Big\}.
\end{eqnarray}
We will assume $ \lambda^{(TIN)}\le\lambda\le\overline{\lambda}$. We need to solve the inequalities $\rho^{(TIN)}_i(r)<\min\{1,r-1\}$ simultaneously for $i=1,2$ in terms of $r$. In order to solve $\rho^{(TIN)}_i(r)<\min\{1,r-1\}$, we apply Lemma~1 with $a=C^*_i$ and $b=C_i$. The three cases $\lambda<\min\{b,\frac{a}{2}\}$, $\frac{a}{2}\leq \lambda<b$ and $b\leq \lambda<\frac{a}{2}$ must be addressed separately for $i=1,2$. This gives a total of $3\times 3=9$ cases to consider. Since $ \lambda\ge\lambda^{(TIN)}$, then at least one of the conditions $\lambda\geq C_1$ and $\lambda\ge C_2$ must hold. This reduces the number of cases to five. In Appendix~D, we study these five cases in detail. The results are captured in the next proposition:
\newpage
\begin{proposition}
Consider the Gaussian interference channel in (\ref{gic}) where the receivers treat interference as noise. Then 
\begin{eqnarray}
\epsilon^{(TIN)}(\lambda)=\kappa \frac{\lambda-C_{i}}{C^*_{i}-2C_{i}}
\end{eqnarray}
where the index $i$ is such that exactly one of the following sets of conditions holds:
\begin{enumerate}
  \item $C_i\leq\lambda<\frac{C_i^*}{2}$ and $\lambda<\min\{C_{i'},\frac{C^*_{i'}}{2}\}$.
  \item $C_i\leq\lambda<\frac{C_i^*}{2}$, $\frac{C^*_{i'}}{2}\leq \lambda<C_{i'}$ and $\frac{C_i^*-2C_i}{C_i^*-C_i-\lambda}<\frac{C_{i'}^*-2C_{i'}}{C_{i'}^*-C_{i'}-\lambda}$.
  \item $C_1\leq\lambda<\frac{C_1^*}{2}$, $C_2\leq\lambda<\frac{C_2^*}{2}$ and $\frac{C_i^*-2C_i}{C_i^*-C_i-\lambda}\ge\frac{C_{i'}^*-2C_{i'}}{C_{i'}^*-C_{i'}-\lambda}$.
\end{enumerate}
\end{proposition}
\item Assume the receivers decode and cancel interference. Then $\lambda^{(DI)}$ in (\ref{l_di}) is given by 
\begin{eqnarray}
\lambda^{(DI)}=\min\Big\{\mathsf{C}(P_1), \mathsf{C}(P_2),\mathsf{C}\Big(\frac{c_1P_2}{1+P_1}\Big), \mathsf{C}\Big(\frac{c_2P_1}{1+P_2}\Big)\Big\}.
\end{eqnarray} 
We will assume $\lambda^{(DI)}\le\lambda\le\overline{\lambda}$. 
  Since $C_i^*=C_{i,i'}$, we need to redefine $\rho_i^{(DI)}(r)=\frac{\lambda r-\tilde{C}_{i}}{\tilde{C}_i^*-\tilde{C_i}}$ as explained in Remark~3 at the end of Section~IV. Then successful communication is guaranteed for user~$i$ if $\rho_i^{(DI)}(r)<\min\{1,r-1\}$ and $r<\frac{C_i^*}{\lambda}$ for $i=1,2$.  In order to solve $\rho^{(DI)}_i(r)<\min\{1,r-1\}$ we apply Lemma~1 with $a=\tilde{C}^*_i$ and $b=\tilde{C}_i$. The three cases $\lambda<\min\{b,\frac{a}{2}\}$, $\frac{a}{2}\leq \lambda<b$ and $b\leq \lambda<\frac{a}{2}$ must be addressed separately for $i=1,2$. As $\lambda\geq \lambda^{(DI)}$, at least one of $\lambda\geq \tilde{C}_1$ and $\lambda\geq \tilde{C}_2$ holds. Hence, there are a total of five cases to consider. The next proposition summarizes the results: 
\begin{proposition}
Consider the Gaussian interference channel in (\ref{gic}) where the receivers decode and cancel interference. Then 
\begin{eqnarray}
\epsilon^{(DI)}(\lambda)=\kappa \frac{\lambda-\tilde{C}_{i}}{\tilde{C}^*_{i}-2\tilde{C}_{i}}
\end{eqnarray}
where the index $i$ is such that exactly one of the following sets of conditions holds:
\begin{enumerate}
  \item $\tilde{C}_i\leq\lambda<\frac{\tilde{C}_i^*}{2}$, $\lambda<\min\{\tilde{C}_{i'},\frac{\tilde{C}^*_{i'}}{2}\}$ and $\frac{\tilde{C}_i^*-2\tilde{C}_i}{\tilde{C}_i^*-\tilde{C}_i-\lambda}<\frac{1}{\lambda}\min\{C^*_1,C^*_2\}$.
  \item $\tilde{C}_i\leq\lambda<\frac{\tilde{C}_i^*}{2}$, $\frac{\tilde{C}^*_{i'}}{2}\leq \lambda<\tilde{C}_{i'}$ and $\frac{\tilde{C}_i^*-2\tilde{C}_i}{\tilde{C}_i^*-\tilde{C}_i-\lambda}<\min\{\frac{\tilde{C}_{i'}^*-2\tilde{C}_{i'}}{\tilde{C}_{i'}^*-\tilde{C}_{i'}-\lambda},\frac{C_{i'}^*}{\lambda}\}\leq \frac{C_i^*}{\lambda}$.
  \item $\tilde{C}_1\leq\lambda<\frac{\tilde{C}_1^*}{2}$, $\tilde{C}_2\leq\lambda<\frac{\tilde{C}_2^*}{2}$ and $\frac{\tilde{C}_{i'}^*-2\tilde{C}_{i'}}{\tilde{C}_{i'}^*-\tilde{C}_{i'}-\lambda}\le\frac{\tilde{C}_i^*-2\tilde{C}_i}{\tilde{C}_i^*-\tilde{C}_i-\lambda}\le\frac{1}{\lambda}\min\{C_1^*,C_2^*\}$.
  \end{enumerate}
\end{proposition}
\begin{proof}
The proof is similar to that of Proposition~4.  
\end{proof}
\end{enumerate}
\begin{figure}[t]
  \centering
  \includegraphics[scale=0.5] {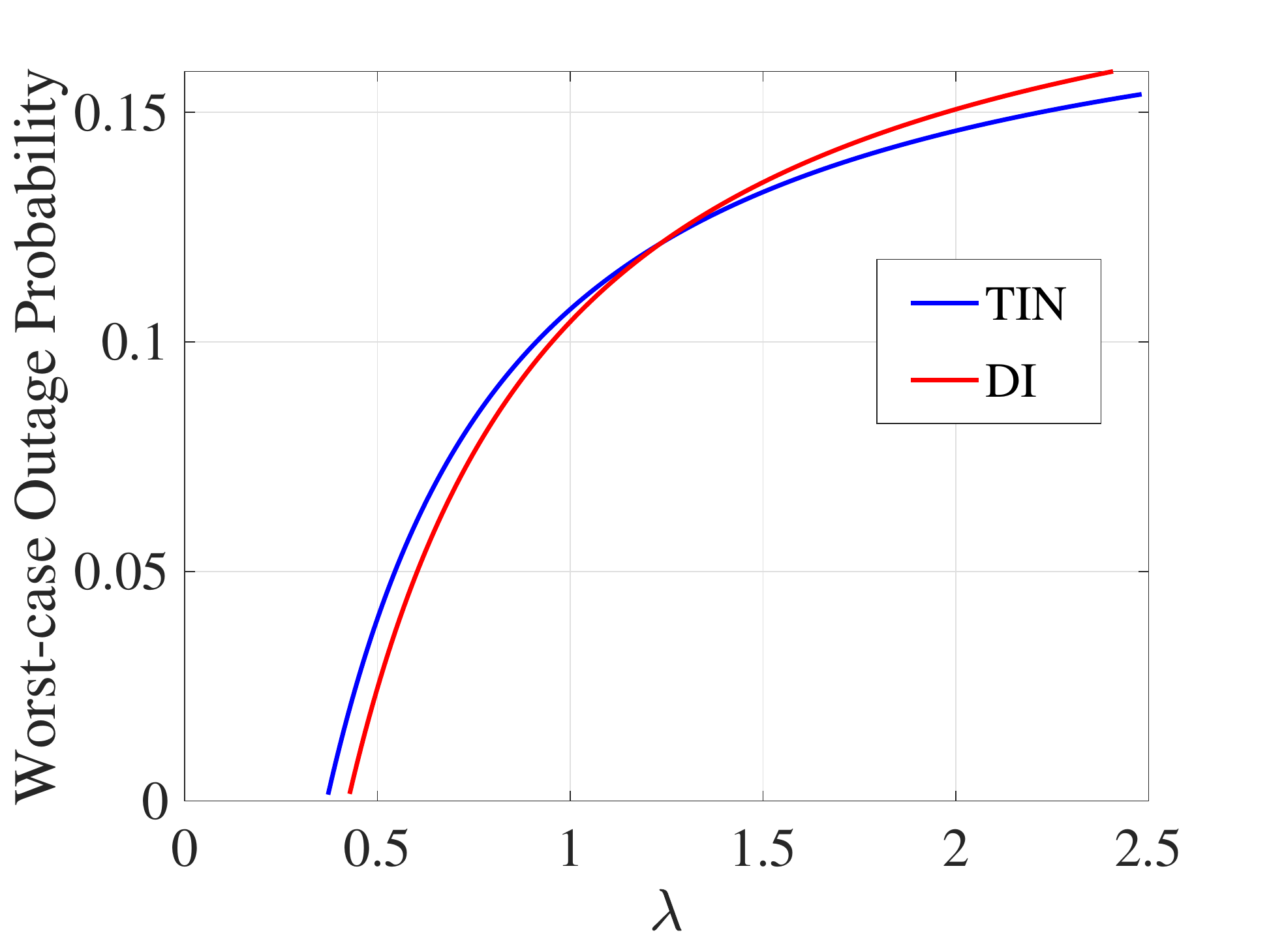}
  \caption{Plots of $\epsilon^{(TIN)}(\lambda)$ and $\epsilon^{(DI)}(\lambda)$ in a Gaussian interference channel with $D=5$, $P_1=P_2=30\,\mathrm{dBW}$, $c_1=0.8$ and $c_2=1.5$.}
  \label{fig1_tin_di}
 \end{figure} 

\textit{Example 3}- Consider a Gaussian interference channel with $D=5$, $P_1=P_2=30\,\mathrm{dBW}$, $c_1=0.8$ and $c_2=1.5$. We get $C_1=0.5845$, $C_2=0.3683$, $C_1^*=C_2^*=4.9836$, $\tilde{C}_1=0.4237$, $\tilde{C}_2=0.6605$, $\tilde{C}_1^*=4.8228$ and $\tilde{C}_2^*=5.2759$. Then $\lambda^{(TIN)}=0.3720$, $\lambda^{(DI)}=0.4279$ and $\overline{\lambda}=4.9836$. Fig.~\ref{fig1_tin_di} presents the plots of the bounds $\epsilon^{(TIN)}(\lambda)$ and $\epsilon^{(DI)}(\lambda)$ in terms of $\lambda$. Proposition~4 allows us to compute $\epsilon^{(TIN)}(\lambda)$ only for  values of $\lambda$ as large as $\min\{\frac{1}{2}C^*_1, \frac{1}{2}C^*_2\}=2.4918$. Similarly, Proposition~5 allows us to compute $\epsilon^{(DI)}(\lambda)$ only for  values of $\lambda$ as large as $\min\{\frac{1}{2}\tilde{C}^*_1, \frac{1}{2}\tilde{C}^*_2\}=2.4114$. We observe that decoding interference outperforms treating interference as noise when $\lambda$ is approximately larger than $1.25$.~$\Diamond$

  \section{Discussion and Concluding Remarks}
  We have studied a two-user IC with gradual data arrival in the absence of feedback. The data sources at the two transmitters were asynchronous and the amount of Tx-Tx asynchrony was unknown to both users. In a scenario where the communication load per user was $n$ bits, we called a number $0<\epsilon<1$ an achievable outage level if the transmission rate per user is $\lambda$ and the probability that the bit-error-rate per user does not vanish is not larger than $\epsilon$. We introduced the tradeoff curve $\epsilon(\lambda)$ which gives the smallest outage level when the rate of data arrival is $\lambda$.  Our main contribution was an upper bound (achievability result) on $\epsilon(\lambda)$. The proposed method of communication was a block transmission scheme where each user transmits a codeword as soon as there are enough bits accumulated in its buffer to represent one. We only studied the performance of random point-to-point encoders and decoders that treat interference as noise or decode interference. It is straightforward to extend the results of the paper to encoders that employ superposition coding (the Han-Kobayashi scheme). Such extensions were omitted in order to maintain the presentation as short as possible.
 \newpage 
 Possible directions for future research include
 \begin{enumerate}
  \item Extending the results to the case where the rates of data arrival $\lambda_1$ and $\lambda_2$ at the two transmitters are not identical. 
  \item Proposing a transmission scheme that outperforms the block transmission scheme in this paper.
  \item Deriving a lower bound (converse result) on $\epsilon(\lambda)$ or $\epsilon(\lambda_1,\lambda_2)$ in general. 
\end{enumerate}

  \section*{Appendix~A; Proof of Communication Failure when $\lambda>\overline{\lambda}$ }
 Fix realizations $d_i$ for $\boldsymbol{d}_i$ for $i=1,2$ and assume the encoder-decoder pairs $(f_{i,n}, g_{i,n})$ are such that 
 \begin{eqnarray}
 \label{cind_11}
\lim_{n\to\infty}R_{i,n}=\lambda
\end{eqnarray}
and
\begin{eqnarray}
\label{cind_22}
\lim_{n\to\infty}p_{i,n}^{(b)}(d_1,d_2)=0
\end{eqnarray}
  for $i=1,2$. Since the two sources are independent, it is easy to see that 
  \begin{eqnarray}
 (\boldsymbol{b}_{i,1},\cdots,\boldsymbol{b}_{i,n})\rightarrow (\boldsymbol{x}_{i,t})_{t\in \mathcal{T}_i}\rightarrow (\boldsymbol{y}_{i,t})_{t\in \mathcal{T}_i}\rightarrow (\hat{\boldsymbol{b}}_{i,1},\cdots,\hat{\boldsymbol{b}}_{i,n})
\end{eqnarray}
forms a Markov chain. Then we can use the data processing inequality to write  
 \begin{eqnarray}
 \label{conv_123}
I\big((\boldsymbol{x}_{i,t})_{t\in \mathcal{T}_i};(\boldsymbol{y}_{i,t})_{t\in \mathcal{T}_i}\big)&\stackrel{}{\geq}& I(\boldsymbol{b}_{i,1},\cdots,\boldsymbol{b}_{i,n};  \hat{\boldsymbol{b}}_{i,1},\cdots,\hat{\boldsymbol{b}}_{i,n}) \notag\\
&=&H(\boldsymbol{b}_{i,1},\cdots,\boldsymbol{b}_{i,n})-H(\boldsymbol{b}_{i,1},\cdots,\boldsymbol{b}_{i,n} |  \hat{\boldsymbol{b}}_{i,1},\cdots,\hat{\boldsymbol{b}}_{i,n})\notag\\
&\stackrel{(a)}{\geq}&\sum_{j=1}^n H(\boldsymbol{b}_{i,j})-\sum_{j=1}^n H(\boldsymbol{b}_{i,j}|\hat{\boldsymbol{b}}_{i,j})\notag\\
&\stackrel{(b)}{\geq}& n-\sum_{j=1}^nh_b\big(\Pr(\hat{\boldsymbol{b}}_{i,j}\neq \boldsymbol{b}_{i,j})\big)\notag\\
&\stackrel{}{=}& n-n\sum_{j=1}^n\frac{1}{n}h_b\big(\Pr(\hat{\boldsymbol{b}}_{i,j}\neq \boldsymbol{b}_{i,j})\big)\notag\\
&\stackrel{(c)}{\geq}& n-nh_b\Big(\frac{1}{n}\sum_{j=1}^n\Pr(\hat{\boldsymbol{b}}_{i,j}\neq \boldsymbol{b}_{i,j})\Big)\notag\\
&=&n\big(1-h_b(p_{i,n}^{(b)}(d_1,d_2))\big),
\end{eqnarray}
where $(a)$ is due to $\boldsymbol{b}_{i,1},\cdots, \boldsymbol{b}_{i,n}$ being independent and the fact that the joint entropy is not larger than the sum of marginal entropies, $(b)$ is due to $\boldsymbol{b}_{i,1},\cdots, \boldsymbol{b}_{i,n}$ being Bernoulli with parameter $\frac{1}{2}$ and Fano's inequality where $h_b(\cdot)$ is the binary entropy function and $(c)$ is due to Jensen's inequality and concavity of $h_{b}(\cdot)$. On the other hand,
\begin{eqnarray}
\label{baroone_123}
I\big((\boldsymbol{x}_{i,t})_{t\in \mathcal{T}_i};(\boldsymbol{y}_{i,t})_{t\in \mathcal{T}_i}\big)&\stackrel{(a)}{\leq}&I\big((\boldsymbol{x}_{i,t})_{t\in \mathcal{T}_i};(\boldsymbol{y}_{i,t})_{t\in \mathcal{T}_{i}}|(\boldsymbol{x}_{i',t})_{t\in \mathcal{T}_{i}}\big)\notag\\
&\stackrel{(b)}{\leq}&\sum_{t\in \mathcal{T}_i}I(\boldsymbol{x}_{i,t};\boldsymbol{y}_{i,t}|\boldsymbol{x}_{i',t})\notag\\
&\stackrel{}{\le}&|\mathcal{T}_i|\max_{\pi_1, \pi_2}I(\boldsymbol{x}_i; \boldsymbol{y}_i|\boldsymbol{x}_{i'})\notag\\
&\stackrel{(c)}{=}&|\mathcal{T}_i|\max_{\pi_1, \pi_2}C_{i,i'},
\end{eqnarray}
where $(a)$ is due to the independence of $(\boldsymbol{x}_{i,t})_{t\in \mathcal{T}_i}$ and $(\boldsymbol{x}_{i',t})_{t\in \mathcal{T}_{i}}$, $(b)$ holds  because the channel is memoryless and $C_{i,i'}$ in $(c)$ was defined in (\ref{ciip}). By (\ref{conv_123}) and (\ref{baroone_123}), 
\begin{eqnarray}
(1-h_b(p_{i,n}^{(b)}(d_1,d_2))R_{i,n}\leq\max_{\pi_1, \pi_2}C_{i,i'},
\end{eqnarray}
where we replaced $\frac{n}{|\mathcal{T}_i|}$ by $R_{i,n}$. Letting $n$ grow to infinity and using (\ref{cind_11}) and (\ref{cind_22}), we get 
\begin{eqnarray}
\lambda\leq \max_{\pi_1, \pi_2}C_{i,i'}.
\end{eqnarray}

\begin{figure}[t]
  \centering
  \includegraphics[scale=0.7] {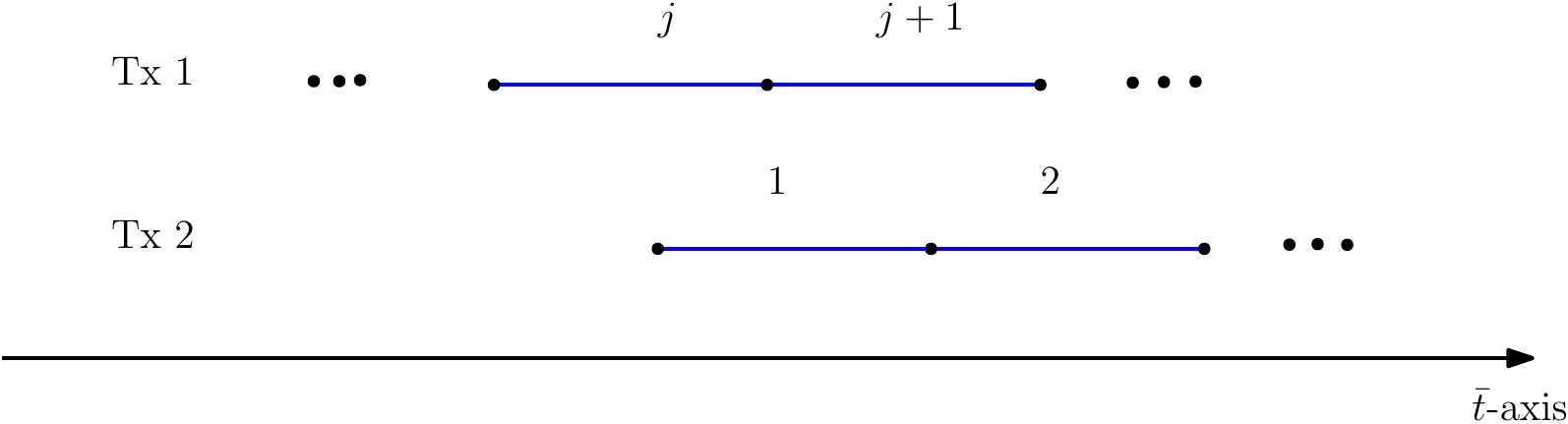}
  \caption{Communication in a two-user IC with gradual data arrival with $r<1$. The blue intervals identify the transmission span of a codeword along the $\bar{t}$-axis. There are no gaps between consecutively transmitted codewords when $r<1$.  The first codeword of Tx~2 is transmitted sometime during the transmission of the $j^{th}$ codeword of Tx~1.  }
  \label{fig1_nogap}
 \end{figure}

 \section*{Appendix~B; Supremacy of $r>1$ in contrast to $r<1$}
 Throughout this appendix we denote $\frac{\delta}{\theta_N}$ by $\delta'$ for notational simplicity. We assume Tx~1 becomes active earlier than Tx~2, i.e., $d_1<d_2$. If $r<1$, there will be no gap between consecutively transmitted codewords. If $\delta'>N$, all codewords are received in the absence of interference. In this case, the condition $R_c<C_i^*$ guarantees successful decoding at Rx~$i$.  Next, assume the first codeword of Tx~2 is transmitted sometime during the transmission of the $j^{th}$ codeword of Tx~1 as shown in Fig.~\ref{fig1_nogap}. Let the receivers treat interference as noise. We consider two cases:  
 \begin{enumerate}
  \item Let $j=N$. This situation occurs when $r+N-1<\delta'+r<r+N$, i.e., 
  \begin{eqnarray}
  \label{far1}
N-1<\delta'<N.
\end{eqnarray} 
The codewords at Rx~$i$ are decoded successfully if 
  \begin{eqnarray}
  \label{ap1}
R_c<(1-\mu)C_i^*+\mu C_i
\end{eqnarray}
where $\mu$ is the fraction of the first codeword of Tx~2 that overlaps with the $N^{th}$ codeword of Tx~1. We can write (\ref{ap1}) as 
\begin{eqnarray}
\label{ap11}
\mu<1-\rho_i(r),
\end{eqnarray}
where $\rho_i(r)$ is defined in (\ref{lola}). For (\ref{ap11}) to be meaningful, we require $\rho_i(r)<1$. Note that 
\begin{eqnarray}
\label{ap2}
\mu=r+N-(\delta'+r)= N-\delta'.
\end{eqnarray}
By (\ref{ap11}) and (\ref{ap2}), 
\begin{eqnarray}
\label{far2}
\delta'>N-1+\rho_i(r).
\end{eqnarray}
Intersecting the inequalities in (\ref{far1}) and (\ref{far2}), we get
\begin{eqnarray}
N-1+\rho_i(r)<\delta'<N.
\end{eqnarray}
  \item Let $1\leq j\leq N-1$. This occurs when $\delta'+r<r+N-1$, i.e., 
  \begin{eqnarray}
\delta'<N-1.
\end{eqnarray}
 The codewords at Rx~$i$ are decoded successfully if 
 \begin{eqnarray}
R_c<C_i.
\end{eqnarray}
This is due to the fact that at least one of the codewords is received fully in the presence of interference. 
\end{enumerate}
We conclude that all codewords are decoded successfully if exactly one of the three sets of conditions 
\begin{eqnarray}
\delta'>N,\,\,\,\rho_i(r)<1,
\end{eqnarray} 
\begin{eqnarray}
N-1+\rho_i(r)<\delta'<N,\,\,\,\rho_i(r)<1
\end{eqnarray}
and
\begin{eqnarray}
\delta'<N-1,\,\,\,\rho_i(r)<0
\end{eqnarray}
hold. Let $\mathds{1}_{\rho_i(r)<1}$ br the indicator variable for $\rho_i(r)<1$. Then 
\begin{eqnarray}
\Pr(\mathcal{O}_i)&\leq& 1-\Pr(\boldsymbol{\delta}>N\theta_N)\mathds{1}_{\rho_i(r)<1}\notag\\&&-\Pr\big((N-1+\rho_i(r))\theta_N<\boldsymbol{\delta}<N\theta_N\big)\mathds{1}_{\rho_i(r)<1}-\Pr\big(\boldsymbol{\delta}<(N-1)\theta_N\big)\mathds{1}_{\rho_i(r)<0}\notag\\
&=&1-\Pr\big(\boldsymbol{\delta}>(N-1+\rho_i(r))\theta_N\big)\mathds{1}_{\rho_i(r)<1}-\Pr\big(\boldsymbol{\delta}<(N-1)\theta_N\big)\mathds{1}_{\rho_i(r)<0}\notag\\
&=&1-\big(1-F_{\boldsymbol{\delta}}\big((N-1+\rho_i(r))\theta_N\big)\big)\mathds{1}_{\rho_i(r)<1}-F_{\boldsymbol{\delta}}\big((N-1)\theta_N\big)\mathds{1}_{\rho_i(r)<0}.
\end{eqnarray}
Replacing $\theta_N=\frac{1}{Nr\lambda}$ and using the continuity of $F_{\boldsymbol{\delta}}(\cdot)$ in (\ref{cdf_11}), we get 
\begin{eqnarray}
\lim_{N\to\infty}\Pr(\mathcal{O}_i)&\leq& 1-(1-F_{\boldsymbol{\delta}}(1/r\lambda))\mathds{1}_{\rho_i(r)<1}-F_{\boldsymbol{\delta}}\big(1/(r\lambda)\big)\mathds{1}_{\rho_i(r)<0}\notag\\
&=&\left\{\begin{array}{cc}
   0   &  \rho_i(r)<0  \\
    F_{\boldsymbol{\delta}}(\frac{1}{r\lambda}) &   0\leq \rho_{i}(r)<1\\
    1 & \rho_i(r)\geq1
\end{array}\right..
\end{eqnarray} 
Letting $r$ approach $1$ from below, 
\begin{eqnarray}
\lim_{r\to1^-}\lim_{N\to\infty}\Pr(\mathcal{O}_i)\leq \left\{\begin{array}{cc}
   0   &  \lambda\le C_i  \\
    F_{\boldsymbol{\delta}}(\frac{1}{\lambda}) &   C_i<\lambda\leq C_i^*\\
    1& \lambda>C_i^*
\end{array}\right..
\end{eqnarray}
By (\ref{cdf_11}), 
\begin{eqnarray}
F_{\boldsymbol{\delta}}(1/\lambda)=\left\{\begin{array}{cc}
   \frac{1}{\lambda D}(2-\frac{1}{\lambda D})   & \lambda D>1   \\
    1  &   \lambda D\le1
\end{array}\right.=\frac{1}{2}\kappa,
\end{eqnarray}
where $\kappa$ is defined in (\ref{capp}). One may  compare $\frac{1}{2}\kappa$ with the upper bound $\kappa \beta_i$ obtained in (\ref{bet_1}) when $r>1$. By Remark~4 after the statement of Proposition~3, $\beta_i<\frac{1}{2}$ and hence, the probability of outage is smaller when $r>1$.

\section*{Appendix~C; Proof of Proposition~\ref{prop_3}}
Throughout this appendix we denote $\frac{\delta}{\theta_N}$ by $\delta'$ for notational simplicity. We assume Tx~1 becomes active earlier than Tx~2, i.e., $d_1<d_2$. More precisely, let Tx~2 start its activity after Tx~1 sends its $j^{th}$ codeword and before it sends its $(j+1)^{st}$ codeword for some $1\leq j\leq N-1$. 
Recall the activity intervals along the $\overline{t}$-axis given in~(\ref{act_int_as}).  The interference pattern on the transmitted codewords depends on if the interval $\mathcal{I}_{2,1}$ representing the first codeword of Tx~2 intersects with both, exactly one or none of the intervals $\mathcal{I}_{1,j}$ and/or $\mathcal{I}_{1,j+1}$, i.e., the intervals representing the $j^{th}$ and $(j+1)^{st}$ codewords of Tx~1. When Rx~$i$ treats interference as noise, its codewords are decoded successfully~if 
\begin{eqnarray}
R_c<(1-\mu_j) C_i^*+\mu_j C_i,
\end{eqnarray}
or equivalently, 
 \begin{eqnarray}
 \label{boo_ma}
\mu_j<1-\rho_i(r),
\end{eqnarray}
 where $\rho_i(r)=\frac{R_c-C_i}{C_i^*-C_i}$ is defined in (\ref{lola}) and  
\begin{eqnarray}
\label{lola11}
\mu_j:=|\mathcal{I}_{2,1}\cap\mathcal{I}_{1,j}|+|\mathcal{I}_{2,1}\cap\mathcal{I}_{1,j+1}|
\end{eqnarray}
is the fraction of $\mathcal{I}_{2,1}$ that is received in the presence of interference. Here, $|\mathcal{I}|$ denotes the length of an interval $\mathcal{I}$.  Note that for (\ref{boo_ma}) to be meaningful it is necessary that $\rho_i(r)<1$. When Rx~$i$ decodes interference, its codewords are decoded successfully if both conditions
\begin{eqnarray}
R_c<(1-\mu_j) \tilde{C}_i^*+\mu_j \tilde{C}_{i},\,\,\,\,\,R_c<(1-\mu_j) C_i^*+\mu_j C_{i,i'}
\end{eqnarray}
hold. Then we obtain (\ref{boo_ma}) with $\rho_i(r)=\max\big\{\frac{R_c-\tilde{C}_{i}}{\tilde{C}^*_i-\tilde{C}_{i}},\frac{R_c-C_{i,i'}}{C_i^*-C_{i,i'}}\big\}$ defined in (\ref{lola}) for receivers that decode interference. Once again we require~$\rho_i(r)<1$ for (\ref{boo_ma}) to be meaningful. 
 
 We continue by considering four cases: 
\begin{enumerate}
  \item Let $\mathcal{I}_{2,1}\cap \mathcal{I}_{1,j}\neq \emptyset$ and $\mathcal{I}_{2,1}\cap \mathcal{I}_{1,j+1}= \emptyset$. This situation is shown in Fig.~\ref{ref333_111}. It occurs when $jr<r+\delta'<jr+1<r+1+\delta'<(j+1)r$. This is equivalent to
  \begin{eqnarray}
   \label{era_11}
(j-1)r<\delta'<\min\{jr-1,(j-1)r+1\}.
\end{eqnarray}
Moreover, $\mu_j=(jr+1)-(r+\delta')=(j-1)r+1-\delta'$. Putting this in (\ref{boo_ma})~yields 
\begin{eqnarray}
\label{era_22}
\delta'>(j-1)r+\rho_i(r).
\end{eqnarray}
Intersecting the intervals on $\delta'$ in (\ref{era_11}) and (\ref{era_22}), we get 
\begin{eqnarray}
\label{ker_11}
(j-1)r+\rho_i(r)<\delta'<\min\{jr-1,(j-1)r+1\}.
\end{eqnarray}
This interval is nonempty if and only if $\rho_i(r)<r-1$.
 \begin{figure}[t]
\centering
\subfigure[]{
\includegraphics[scale=0.6]{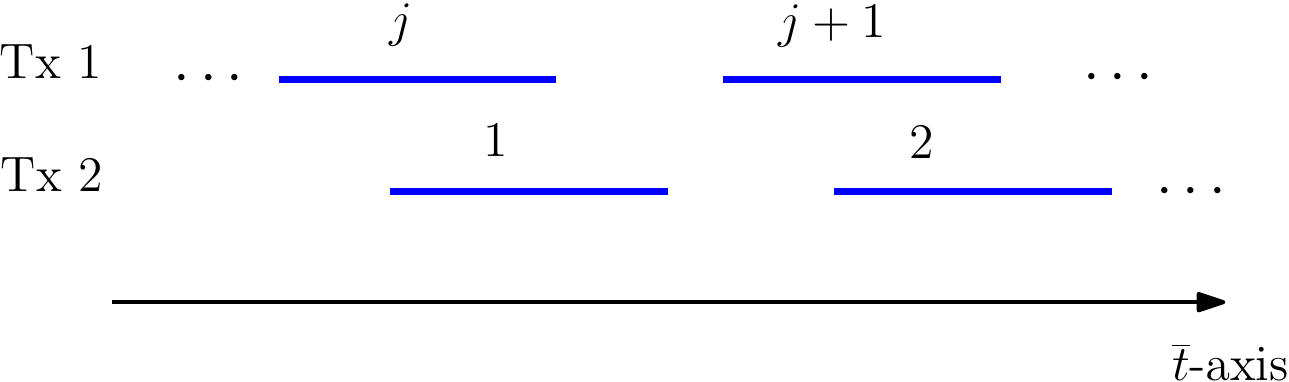}
\label{ref333_111}
}
\subfigure[]{
\includegraphics[scale=0.6]{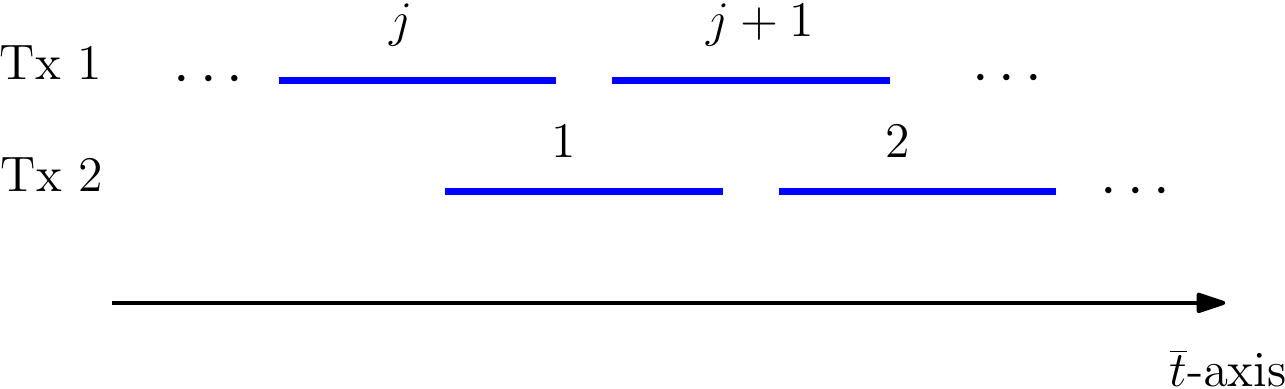}
\label{ref333_222}
}
\subfigure[]{
\includegraphics[scale=0.6]{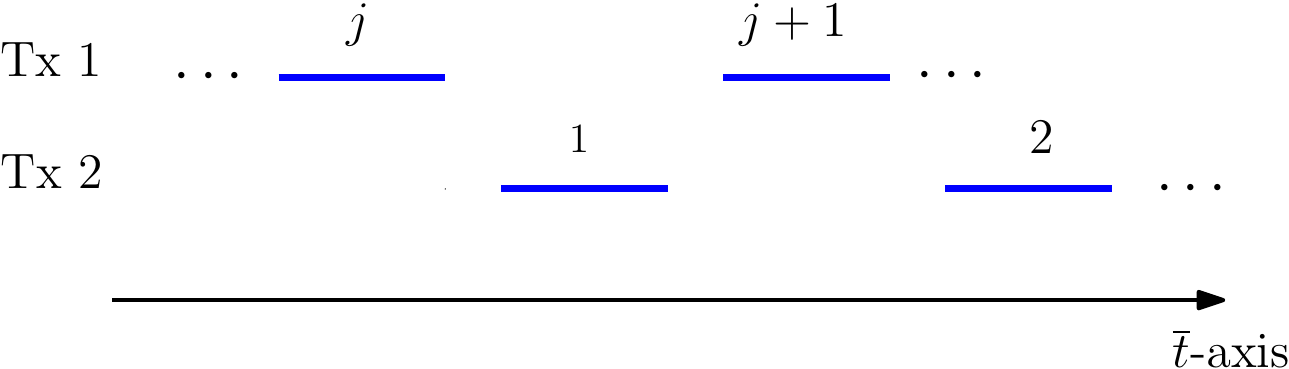}
\label{ref333_333}
}
\subfigure[]{
\includegraphics[scale=0.6]{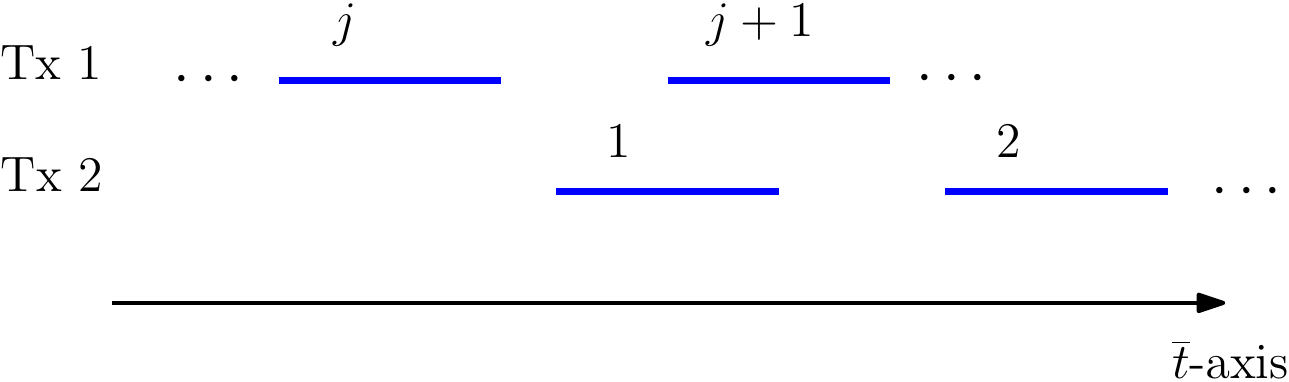}
\label{ref333_444}
}

\label{fig:subfigureExample}
\caption[Optional caption for list of figures]{Tx~2 starts its activity after Tx~1 sends its $j^{th}$ codeword and before it sends its $(j+1)^{st}$ codeword for some $1\leq j\leq N$. The interference pattern on the transmitted codewords depends on if the interval $\mathcal{I}_{2,1}$ representing the first codeword of Tx~2 intersects with intervals $\mathcal{I}_{1,j}$ and/or $\mathcal{I}_{1,j+1}$, i.e., the $j^{th}$ and $(j+1)^{st}$ codewords of Tx~1. The four possibilities are depicted in panels (a) to (d).}
\label{yave_33}
\end{figure}
  \item Let $\mathcal{I}_{2,1}\cap \mathcal{I}_{1,j}\neq\emptyset$ and $\mathcal{I}_{2,1}\cap \mathcal{I}_{1,j+1}\neq \emptyset$. This situation is shown in Fig.~\ref{ref333_222}. It occurs when $jr<r+\delta'<jr+1<(j+1)r<r+1+\delta'$. This simplifies to
  \begin{eqnarray}
  \label{ker_22}
jr-1<\delta'<(j-1)r+1.
\end{eqnarray}
Looking at Fig.~\ref{ref333_222}, we have $\mu_j=[(jr+1)-(r+\delta')]+[(r+1+\delta')-(j+1)r]=2-r$. Putting this in (\ref{boo_ma}), we obtain the inequality $\rho_i(r)<r-1$. We note that this case occurs only when $1<r<2$.
  \item  Let $\mathcal{I}_{2,1}\cap \mathcal{I}_{1,j}=\emptyset$ and $\mathcal{I}_{2,1}\cap \mathcal{I}_{1,j+1}=\emptyset$. This situation is shown in Fig.~\ref{ref333_333}. It occurs when both $jr+1<r+\delta'$ and $r+1+\delta'<(j+1)r$ hold, i.e.,  
 \begin{eqnarray}
 \label{hunbun}
 (j-1)r+1<\delta'<jr-1.
\end{eqnarray}
 We note that this case occurs only when $r>2$. 
  \item Let $\mathcal{I}_{2,1}\cap \mathcal{I}_{1,j}=\emptyset$ and $\mathcal{I}_{2,1}\cap \mathcal{I}_{1,j+1}\neq \emptyset$. This situation is shown in Fig.~\ref{ref333_444}. It occurs when $jr+1<r+\delta'<(j+1)r<r+1+\delta'$ which simplifies to
  \begin{eqnarray}
  \label{era_33}
\max\{jr-1,(j-1)r+1\}<\delta'<jr,
\end{eqnarray}
Looking at Fig.~\ref{ref333_444}, we have $\mu_j=(r+1+\delta')-(j+1)r=1+\delta'-jr$. Putting this in (\ref{boo_ma}), we get the inequality 
\begin{eqnarray}
\label{era_44}
\delta'<jr-\rho_i(r).
\end{eqnarray}
Intersecting the intervals on $\delta'$ given in (\ref{era_33}) and (\ref{era_44}), we get 
\begin{eqnarray}
\label{ker_33}
\max\{jr-1,(j-1)r+1\}<\delta'<jr-\rho_i(r).
\end{eqnarray}
This interval is nonempty if and only if $\rho_i(r)<r-1$.
\end{enumerate}
To recap, successful communication occurs if $\rho_i(r)<1$, $\rho_i(r)<r-1$ and at least one of the inequalities in (\ref{ker_11}), (\ref{ker_22}), (\ref{hunbun}) or (\ref{ker_33}) hold. Uniting the corresponding intervals, we get $\delta'\in \mathcal{A}(j)$ where $\mathcal{A}(j)$ is given in (\ref{khia_11}). 

Next, we address the case $j=N$. The first codeword of Tx~2 intersects the $N^{th}$ codeword of Tx~1 if $Nr<r+\delta'<Nr+1$, i.e., 
\begin{eqnarray}
\label{pank}
(N-1)r<\delta'<(N-1)r+1.
\end{eqnarray}
 We also have $\mu_j=Nr+1-(r+\delta')=(N-1)r+1-\delta'$. Then the condition in (\ref{boo_ma}) becomes $\delta'>(N-1)r-\rho_i(r)$. Combining this inequality with the ones in (\ref{pank}), we get $(N-1)r-\rho_i(r)<\delta'<(N-1)r+1$.
 
 Finally, when $\delta'>(N-1)r+1$, the two users do not interfere at all. In this case, the single condition $\rho_i(r)<1$ guarantees successful communication.

\section*{Appendix~D;~ Proof of Proposition~\ref{prop_7}}
For notational simplicity, we will denote $\rho_i(r)$ by $\rho$ and $\theta_N$ by $\theta$. Let us write $\mathcal{A}_{i,j}=(a_j,b_j)$ for $1\leq j\leq N-1$ and $\mathcal{A}_{i,N}=(a_N,\infty)$ where
\begin{eqnarray}
a_j=(j-1)r+\rho,\,\,\,\,\,b_j=jr-\rho.
\end{eqnarray}
 If $\chi_{1}(i)=1$, then $\rho<\min\{1,r-1\}$ and we get  
\begin{eqnarray}
\label{rav_000}
0<a_1<b_1<a_2<b_2<\cdots<a_{N-1}<b_{N-1}<a_N.
\end{eqnarray}
Then 
\begin{eqnarray}
\label{rav_111}
\Pr\Big(\boldsymbol{\delta}\in \bigcup_{j=1}^{N-1}\theta\mathcal{A}_{i,j}\Big)=\sum_{j=1}^{N-1}\big(F_{\boldsymbol{\delta}}(\theta b_j)-F_{\boldsymbol{\delta}}(\theta a_j)\big),
\end{eqnarray}
where $F_{\boldsymbol{\delta}}(\cdot)$ is given in (\ref{cdf_11}). The right side of (\ref{rav_111}) can be further simplified depending on where the number $D/\theta$ stands among the terms of the sequence in (\ref{rav_000}). Three cases occur: 
\begin{enumerate}
  \item Let $D/\theta<a_1$. Then the right side of (\ref{rav_111}) vanishes due to $F_{\boldsymbol{\delta}}(\theta a_j)=F_{\boldsymbol{\delta}}(\theta b_j)=1$ for every $j$. 
  \item Let $a_m< D/\theta<b_m$ for some $1\leq m\leq N-1$. Then (\ref{rav_111}) becomes
\begin{eqnarray}
\label{rav_222}
\Pr\Big(\boldsymbol{\delta}\in \bigcup_{j=1}^{N-1}\theta\mathcal{A}_{i,j}\Big)&=&\sum_{j=1}^{m-1}\big(F_{\boldsymbol{\delta}}(\theta b_j)-F_{\boldsymbol{\delta}}(\theta a_j)\big) +1-F_{\boldsymbol{\delta}}(\theta a_m)\notag\\
&=&\sum_{j=1}^{m-1}\Big(\frac{\theta b_j}{D}\big(2-\frac{\theta b_j}{D}\big)-\frac{\theta a_j}{D}\big(2-\frac{\theta a_j}{D}\big)\Big) +1-\frac{\theta a_m}{D}\big(2-\frac{\theta a_m}{D}\big)\notag\\
&=&\sum_{j=1}^{m-1}\frac{\theta}{D}\big(2-\frac{\theta}{D}(a_j+b_j)\big)(b_j-a_j) +\big(1-\frac{\theta a_m}{D}\big)^2\notag\\
&\stackrel{(a)}{=}&(r-2\rho)\frac{\theta}{D}\sum_{j=1}^{m-1}\big(2-\frac{r\theta}{D}(2j-1)\big)+\big(1-\frac{\theta a_m}{D}\big)^2\notag\\
&=&(r-2\rho)\frac{\theta}{D}\Big(2(m-1)-\frac{r\theta}{D}\sum_{j=1}^{m-1}(2j-1)\Big)+\big(1-\frac{\theta a_m}{D}\big)^2\notag\\
&\stackrel{(b)}{=}&(r-2\rho)\frac{\theta}{D}\big(2(m-1)-\frac{r\theta}{D}(m-1)^2\big)+\big(1-\frac{\theta a_m}{D}\big)^2,
\end{eqnarray}
where $(a)$ is due to $a_j+b_j=(2j-1)r$ and $b_j-a_j=r-2\rho$ and $(b)$ is due to $\sum_{j=1}^{m-1}(2j-1)=(m-1)^2$.
  \item Let $b_m< D/\theta<a_{m+1}$ for some $1\leq m\leq N-1$. Following similar lines of reasoning as in~(\ref{rav_222}),  we get
  \begin{eqnarray}
\Pr\Big(\boldsymbol{\delta}\in \bigcup_{j=1}^{N-1}\theta\mathcal{A}_{i,j}\Big)=\sum_{j=1}^{m}\big(F_{\boldsymbol{\delta}}(\theta b_j)-F_{\boldsymbol{\delta}}(\theta a_j)\big)=(r-2\rho)\frac{\theta}{D}\big(2m-\frac{r\theta}{D}m^2\big).
\end{eqnarray}
\item Let $D/\theta>a_N$. Then we obtain the expression on the right side of (\ref{rav_222}) with $m$ replaced by $N$. 
\end{enumerate}
We continue by interpreting the integer $m$ mentioned in above. The condition $b_m<D/\theta<a_{m+1}$ is equivalent to 
\begin{eqnarray}
D/\theta-\rho<mr<D/\theta+\rho
\end{eqnarray} 
and the condition $a_m<D/\theta<b_{m}$ is equivalent to 
\begin{eqnarray}
D/\theta+\rho<mr<D/\theta-\rho+r.
\end{eqnarray} 
Hence, $mr$ is the (unique) multiple of $r$ that lies between the two numbers $D/\theta-\rho$ and $D/\theta-\rho+r$ concluding that $m=\lceil (D/\theta-\rho)/r\rceil$. Finally, the cases $D/\theta<a_1$ and $D/\theta>a_N$ correspond to $m=0$ and $m\geq N$, respectively.     

\section*{Appendix E; Proof of Proposition~4}
All we need to determine is the quantity $r_0$ defined in Theorem~(\ref{thm1}). In what follows, we invoke Lemma~1 repeatedly. 
\begin{enumerate}
  \item Assume $\lambda<\min\{C_1,\frac{C_1^*}{2}\}$ and $C_2\leq \lambda<\frac{C_2^*}{2}$. Then the solution to $\rho_1(r)<\min\{1,r-1\}$ is 
  \begin{eqnarray}
  \label{ionic_1}
1<r<\frac{C_1^*}{\lambda}
\end{eqnarray}
and the solution to $\rho_2(r)<\min\{1,r-1\}$ is 
\begin{eqnarray}
\label{ionic_2}
\frac{C_2^*-2C_2}{C_2^*-C_2-\lambda}<r<\frac{C^*_2}{\lambda},
\end{eqnarray}
where $1\le\frac{C_2^*-2C_2}{C_2^*-C_2-\lambda}<2$. Since $\frac{C_1^*}{\lambda},\frac{C_2^*}{\lambda}>2$, the intersection of the two intervals in (\ref{ionic_1}) and (\ref{ionic_2}) is given by $\frac{C_2^*-2C_2}{C_2^*-C_2-\lambda}<r<\frac{1}{\lambda}\min_{i=1,2} C_i^*$. Hence, $r_0$ is given by 
\begin{eqnarray}
r_0=\frac{C_2^*-2C_2}{C_2^*-C_2-\lambda}.
\end{eqnarray}
It is easy to see that $\beta_1(r_0)<\beta_2(r_0)$. Then
 \begin{eqnarray}
\epsilon^{(TIN)}(\lambda)&=&\kappa \beta_2(r_0)\notag\\
&=&\kappa \frac{\lambda-\frac{C_2}{r_0}}{C_2^*-C_2}\notag\\
&=&\frac{\kappa}{C_2^*-C_2}\bigg( \lambda-\frac{C_2}{\frac{C_2^*-2C_2}{C_2^*-C_2-\lambda}}\bigg)\notag\\
&=&\kappa \frac{\lambda-C_2}{C^*_2-2C_2}
\end{eqnarray}
   \item Assume $\frac{C_1^*}{2}\le\lambda<C_1$ and $C_2\leq \lambda<\frac{C_2^*}{2}$. Then the solution to $\rho_1(r)<\min\{1,r-1\}$ is 
  \begin{eqnarray}
  \label{ionic_3}
1<r<\frac{C_1^*-2C_1}{C_1^*-C_1-\lambda}
\end{eqnarray}
where $1<\frac{C_1^*-2C_1}{C_1^*-C_1-\lambda}\le2$ and the solution to $\rho_2(r)<\min\{1,r-1\}$ is 
\begin{eqnarray}
\label{ionic_4}
\frac{C_2^*-2C_2}{C_2^*-C_2-\lambda}<r<\frac{C^*_2}{\lambda},
\end{eqnarray}
where $1\le\frac{C_2^*-2C_2}{C_2^*-C_2-\lambda}<2$. The intersection of the two intervals in (\ref{ionic_3}) and (\ref{ionic_4}) is nonempty if and only if $\frac{C_2^*-2C_2}{C_2^*-C_2-\lambda}<\frac{C_1^*-2C_1}{C_1^*-C_1-\lambda}$ and is given by 
\begin{eqnarray}
\frac{C_2^*-2C_2}{C_2^*-C_2-\lambda}<r<\frac{C_1^*-2C_1}{C_1^*-C_1-\lambda}.
\end{eqnarray}
It follows that 
\begin{eqnarray}
r_0=\frac{C_2^*-2C_2}{C_2^*-C_2-\lambda}.
\end{eqnarray}
It is easy to check that $\beta_1(r_0)<\beta_2(r_0)$. Then 
\begin{eqnarray}
\epsilon^{(TIN)}(\lambda)=\kappa \beta_2(r_0)=\kappa \frac{\lambda-C_2}{C^*_2-2C_2}.
\end{eqnarray}
  \item Assume $\lambda<\min\{C_2,\frac{C_2^*}{2}\}$ and $C_1\leq \lambda<\frac{C_1^*}{2}$. This is similar to case (a) in above where the positions of the indices ``1'' and ``2'' are interchanged. We get 
  \begin{eqnarray}
r_0=\frac{C_1^*-2C_1}{C_1^*-C_1-\lambda}
\end{eqnarray}
and 
\begin{eqnarray}
\epsilon^{(TIN)}(\lambda)=\kappa \beta_1(r_0)=\kappa \frac{\lambda-C_1}{C^*_1-2C_1}.
\end{eqnarray}
\item Assume $\frac{C_2^*}{2}\le\lambda<C_2$ and $C_1\leq \lambda<\frac{C_1^*}{2}$. This is similar to case (b) in above where the positions of the indices ``1'' and ``2'' are interchanged. There is a solution if and only if $\frac{C_1^*-2C_1}{C_1^*-C_1-\lambda}<\frac{C_2^*-2C_2}{C_2^*-C_2-\lambda}$. We get 
 \begin{eqnarray}
r_0=\frac{C_1^*-2C_1}{C_1^*-C_1-\lambda}
\end{eqnarray} 
and 
\begin{eqnarray}
\epsilon^{(TIN)}(\lambda)=\kappa \beta_1(r_0)=\kappa \frac{\lambda-C_1}{C^*_1-2C_1}.
\end{eqnarray}
\item Finally, assume $C_1\leq \lambda<\frac{C_1^*}{2}$ and $C_2\leq \lambda<\frac{C_2^*}{2}$. The solution to $\rho_i(r)<\min\{1,r-1\}$ is given by $\frac{C_i^*-2C_i}{C_i^*-C_i-\lambda}<r<\frac{C^*_i}{\lambda}$ for both $i=1,2$. The intersection of these two intervals is nonempty and is given by $\max_{i=1,2}\frac{C_i^*-2C_i}{C_i^*-C_i-\lambda}<r<\frac{1}{\lambda}\min_{i=1,2}C^*_i$.  Then  
\begin{eqnarray}
\label{gbc_1}
r_0=\max_{i=1,2}\frac{C_i^*-2C_i}{C_i^*-C_i-\lambda}.
\end{eqnarray}
Let $i=i_0$ be the value for the index $i\in\{1,2\}$ that achieves the maximum in~(\ref{gbc_1}). Then it is easy to see that $\max\{\beta_1(r_0),\beta_2(r_0)\}=\beta_{i_0}(r_0)$ and we get 
\begin{eqnarray}
\epsilon^{(TIN)}(\lambda)=\kappa \beta_{i_0}(r_0)=\kappa \frac{\lambda-C_{i_0}}{C^*_{i_0}-2C_{i_0}}.
\end{eqnarray}
\end{enumerate}


\end{document}